\documentclass[12pt]{article}
\usepackage[a4paper]{geometry}
\usepackage[utf8]{inputenc}
\usepackage[T1]{fontenc}
\usepackage{lmodern}
\usepackage[british]{babel}
\usepackage{csquotes}
\usepackage[style=numeric]{biblatex}
\bibliography{PhD.bib}
\usepackage{enumitem}
\usepackage{ebproof}
\usepackage{hhline}
\usepackage{ circuitikz }
\usepackage[twopage]{rotating}

\usepackage{amsmath}
\usepackage{amsfonts}
\usepackage{amssymb}
\usepackage{ntheorem}
\usepackage{mathrsfs}
\usepackage{stmaryrd}
\usepackage{proof}
\usepackage{phonetic}
\usepackage{upgreek}
\usepackage{tipa}
\usepackage{graphicx}
\usepackage{yhmath}
\usepackage{mathdots}
\usepackage{MnSymbol}
\usepackage{calc}
\usepackage{layout}
\usepackage{thmtools}
\usepackage{fancyhdr}
\usepackage{longtable}
\usepackage[dvipsnames]{xcolor}

\usepackage{tikz}
\usetikzlibrary{shapes, arrows.meta}
\usepackage{tabularx,colortbl}
\usepackage[twopage]{rotating}
\usepackage{musicography}
\usepackage{multirow}
\usepackage{diagbox}
\definecolor{oxfordblue}{RGB}{4,30,66}
 
\allowdisplaybreaks

\theorembodyfont{\fontshape{rm}}

\newtheorem{definition}{Definition}
\newtheorem{theorem}{Theorem}
\newtheorem{lemma}{Lemma}
\newtheorem{corollary}{Corollary}

\newtheorem{example}{Example}

\renewtheorem{definition}{Definition}
\newtheorem*{proof}{Proof}

\makeatletter
\providecommand{\leftsquigarrow}{%
  \mathrel{\mathpalette\reflect@squig\relax}%
}
\newcommand{\reflect@squig}[2]{%
  \reflectbox{$\m@th#1\rightsquigarrow$}%
}
\makeatother

\usepackage[flushleft]{threeparttable}

\usepackage{setspace}
\usepackage{tcolorbox}
\usepackage{turnstile}

\usepackage{hyperref}
\title{Towards Automated Proof-Theoretic Semantics: Inference-Behaviour Semantics for 3-Dimensional \textbf{K3} and \textbf{LP}}
\author{Sophie Nagler\vspace{0.3cm}\\ Institute for Logic, Language and Computation (ILLC)\\  University of Amsterdam\\ P.O. Box 94242\\  1090 GE Amsterdam\\ The Netherlands\footnote{s.e.nagler@uva.nl} \vspace{0.1cm} \\ \& \vspace{0.1cm} \\ Arché Research Centre\\University of St Andrews\\Scotland\footnote{sen1@st-andrews.ac.uk}\vspace{0.3cm}}
\date{\today}

\begin{document}
\begin{titlepage}
 \maketitle   
\end{titlepage}
\begin{abstract}
Inference-behaviour Semantics (I-bS) is a sequent calculus-based substructural approach to proof-theoretic semantics (P-tS), focussed on modelling the relationships of connective meanings across different logics. This paper provides a case study of how I-bS can be extended to any finitely multi-dimensional sequent calculus. To this end, we give I-bS for the 3-dimensional sequent calculi \textbf{K3} (Strong Kleene) and \textbf{LP} (Logic of Paradox). We find that the \textbf{K3} and \textbf{LP} connectives have the same meaning and that their meaning conservatively extends the meaning of the corresponding classical \textbf{LK} connectives. We situate this result within the wider project of integrating I-bS with MUltlog. The goal is to automate the generation of I-bS for arbitrary multi-valued logics using a MUltlog-style system, thus contributing towards the computational automation of P-tS approaches.
\end{abstract}
\textbf{Keywords:} proof-theoretic semantics $\bullet$ connectives $\bullet$ automation $\bullet$ substructural logics $\bullet$ non-classical logics
\section{Introduction}
The automation of logical reasoning has a long history, from the development of computability theory and decision procedures \cite{church1936unsolvable, turing1936computable} to resolution and automated theorem proving \cite{robinson1965machine}, and large interactive proof assistants such as Rocq (formerly Coq) \cite{bertot2013interactive} and Lean \cite{de2015lean}. However, automating the provision of \textit{proof-theoretic semantics (P-tS)} for such reasoning remains largely an open problem. This paper is part of a larger project whose goal is the computer-assisted generation of proof-theoretic meaning specifications for the connectives of any finite-valued logic. To achieve this goal, it aims to integrate the \textit{inference-behaviour semantics (I-bS)} approach to P-tS \cite{nagler2026measuring, nagler2026inference} into MUltlog, a system that generates a sequent calculus, tableau and natural deduction systems for any finitely-valued propositional or first-order logic \cite{salzer1996multlog}. In this first foundational paper, we extend I-bS to MUltlog's underlying logical machinery of multi-dimensional sequent calculi. We demonstrate this extension via the case study of two 3-valued logics: Strong Kleene logic (\textbf{K3}) \cite{kleene1938notation, kleene1952introduction} and the Logic of Paradox (\textbf{LP}) \cite{priest1979logic}. 

P-tS is a heterogeneous family of approaches to formal semantics in which the meaning an expression is determined by its \textit{use in proofs} \cite{sep-proof-theoretic-semantics}. It operationalises the inferentialist slogan `meaning as use' in terms of \textit{proof rules} in sequent or natural deduction systems. One can think of the proof rules for an expression as formal instructions for its proper use in a context of reasoning. In the case of logical connectives (expressions such as `not', `if', `and'), we call these \textit{use-defining} rules \textit{operational} \cite{gentzen1935auntersuchungen, gentzen1935buntersuchungen}. P-tS thus allows to directly move from proof systems to meanings without needing additional mathematical resources such as a model theory.

However, this promise has not yet been cashed out computationally as P-tS has remained primarily a theoretical and philosophical enterprise. Whilst there is no systematic method for automatically generating semantic clauses for some arbitrary given proof system, there are existing tools for the generation of calculi themselves, in particular, MUltlog.

MUltlog, first presented in \cite{salzer1996multlog}, has been developed since the 1990s based on a series of results by the TU Vienna logic group \cite{baaz1993dual, baaz1993elimination, baaz1993systematic, salzer2000optimal, zach1993proof}. It is a Prolog-based system that takes as input an $n$-valued logic (for some finite $n \in\mathbb{N}$) given as a set of truth tables for the connectives and a selection of designated truth values, and automatically generates tableau, a natural deduction and---the focus of this paper---an $n$-dimensional (also called: $n$-sided) sequent calculus for that logic.  MUltlog, thus, automates the move from (model-theoretic) semantic clauses to proof systems for the finitely-valued case. What MUltlog currently lacks, however, is the converse: a (proof-theoretic) account of connective meanings in the proof systems it generates. Providing such an account and implementing it into MUltlog not only automates the generation of P-tS for finite-valued logics but also links $n$-valued model-theoretic semantic clauses to their proof-theoretic counterpart.

In the inference-behaviour semantics (I-bS) paradigm of P-tS, these semantic clauses take the shape of substructural operational rules. These are obtained by measuring the minimal \textit{inference behaviour} of a connective \#, i.e. the syntactic occurrences of \#-formulae in the proof of \#'s definability relative to a minimal derivability relation. In contrast to other popular approaches to P-tS such as proof-theoretic validity \cite{prawitz1971ideas, prawitz1974idea, prawitz1973towards,  schroeder2006validity} or base-extension semantics \cite{ gheorghiu2024proof, sandqvist2015base}, I-bS models connective meanings in a calculus-independent way. This allows us to compare the meaning of connectives in different calculi, yielding an insightful proof-theoretic analysis of their semantic interrelations such as between classical negation(s) and implication(s) and their intuitionistic and dual-intuitionistic counterparts \cite{nagler2026inference, nagler2026measuring}.\footnote{Moreover, logic-independent connective meanings are critical for modelling meaning-invariant logical pluralism, and for related questions about logical disagreement \cite[e.g.][]{beall2001defending, beall2000logical, beall2006logical}, a major backdrop to the development of I-bS \cite{dicher2016proof, nagler2026measuring,restall2014pluralism}.}

One of the reasons that makes I-bS particularly suitable for this project is that it generally only designates a limited number of connectives as \textit{meaningful} relative to a fixed minimal derivability relation. For instance, of the $10,816$ operational rules definable in standard 2-dimensional sequent calculi using at most two active formulae and premiss sequents, only $21$ turn out to be meaningful in the default minimal derivability relation only containing \textsc{Id(entity)} over primitives and context-differentiating \textsc{Cut} \cite{nagler2026inference}. This greatly limits the possible output space for a MUltlog-style computational implementation. However, I-bS has only been developed for 2-dimensional sequent calculi thus far.

This paper provides a first case study for extending I-bS to $>2$-dimensional calculi, focussing on two characteristic 3-valued logics: Kleene's Strong Theory \textbf{K3} \cite{kleene1952introduction, kleene1938notation} and Priest's Logic of Paradox \textbf{LP} \cite{priest1979logic}, specifically in the form of MUltlog-generated 3-dimensional sequent calculi \cite{multlog2024lp, multlog2024k3}. Whilst extending the I-bS machinery to an added dimension is relatively straightforward, a more sophisticated solution is needed to preserve its signature feature: the comparison of connective meanings across different calculi. In response, we will propose a notion of meaning-extension across dimensions such that an $m$-dimensional connective extends the meaning of an $n$-dimensional connective (for $m>n$) iff the inference behaviour of the latter is (conservatively) preserved in an $n$-dimensional fragment of the former. 

Given this notion, we find that the connectives in the MUltlog-generated 3-dimensional calculus for \textbf{K3} are meaning-extensions of their counterparts in classical \textbf{LK} \cite{gentzen1935auntersuchungen,gentzen1935buntersuchungen}. Moreover, our results show that the \textbf{K3} connectives are meaning-identical to their \textbf{LP}-counterparts, a result akin to \cite{hjortland2013logical}.

To achieve our goals, we first introduce multi-dimensional sequent calculi---in particular, 3-dimensional \textbf{K3}/\textbf{LP}---in \S2, alongside a new double-negative bilateralist reading for multi-dimensional sequents. We then briefly outline the method of I-bS in \S 3. In \S4, we generate our data for the measurement of inference behaviour in the form of conservativity and uniqueness proofs for 3-dimensional \textbf{K3}/\textbf{LP} relative to a minimal derivability relation of 3-dimensional \textsc{Id} and 2-dimensional \textsc{Cut}. Based on this, we present the semantic clauses for \textbf{K3}/\textbf{LP} in \S5 together with the results that the connective meanings in the two calculi are identical and, moreover, meaning-extensions of \textbf{LK}. In \S6, we evaluate our results with respect to our overarching goal of automating I-bS before taking stock in \S 7.
\section{Object of Inquiry: Multi-Dimensional Sequents}
\begin{definition}[syntax]\

\noindent We operate with two different sentential \textit{languages}:
\begin{enumerate}
    \item object language $\mathcal{L}_\textbf{Cal}:=\big\{p_{i}\mid i\in \mathbb{N}\big\}\cup \big\{-, \wedge, \vee, \supset, ), (\big\}$, and
    \item semantic language $\mathcal{L}^S_\textbf{Cal}:=\big\{p_{i}\mid i\in \mathbb{N}\big\}\cup \big\{ \neg, \sim, \otimes, \sqcap, \oplus, \sqcup, \rightarrow, \rightsquigarrow, ), ( \big\}$.
\end{enumerate}
\textit{Formulae} of $\mathcal{L}_\textbf{Cal} \,\big/\,\mathcal{L}^S_\textbf{Cal}$ are defined recursively:
\begin{enumerate}
    \item Every sentential variable $p$ is a formula and \textit{primitive}.
    \item If $A$ is a formula, so is $\#A$, for $\#\in\{-\}\big/\{\neg,\sim\}.$
    \item If $A$ and $B$ are formulae, so is $A\#B$, for $\#\in\{\wedge, \vee,\supset\}\,\big/\,\{\otimes, \sqcap, \oplus, \sqcup, \rightarrow,\rightsquigarrow\}.$
    \item Nothing else is a formula.
\end{enumerate}
The set $PS(A)$ of \textit{primitive sub-formulae} of $A$ is defined recursively: $PS(p):=\{p\}$; $PS(\#A):=PS(A)$, for $\#\in\{-,\neg,\sim \}$; $PS(A\#B):=PS(A)\cup PS(B)$, for $\#\in\{\wedge, \otimes, \sqcap, \vee, \oplus, \sqcup,\supset, \rightarrow,\rightsquigarrow\}$; $PS(\Gamma)=\bigcup_{A\in \Gamma} PS(A)$.

We follow standard notational conventions where Roman majuscules $A,B,C, \dots$ represent arbitrary formulae and Greek majuscules $\Gamma,\Delta,\Sigma,\dots$ represent arbitrary multisets of formulae. We write $\#, \natural$ for arbitrary connectives, and $\texttt{\#}$ $(\na)$ for arbitrary formulae in which $\#$ $(\natural)$ is the main connective.
\end{definition}

\begin{definition}[$n$-dimensional sequent]\

    \noindent 
    For some finite $n\in \mathbb{N}$, an $n$-dimensional \textit{sequent} $\mathfrak{s}$ is an $n$-tuple
    \[\Gamma_1|\Gamma_2|\Gamma_3|\dots|\Gamma_n\]
    such that $\mathcal{D}=\{d\in \mathbb{N}|1\leq d\leq n\}$ is the set of dimensions of that sequent and $\Gamma_{d\in \mathcal{D}}$ is a finite multiset of formulae. We call $\Gamma_d$ the $d$th \textit{component} of $\mathfrak{s}$.

    We use component-wise sequent notation from \cite{baaz1993dual,baaz1993elimination,baaz1993systematic, zach1993proof}, in the variant of \cite{nagler2026inference}:
    \begin{itemize}
    \singlespacing
    \item For some $\mathcal{D}^\prime\subseteq \mathcal{D}$, $[\mathcal{D}^\prime:\Delta]$ is a sequent such that $\Gamma_d=
    \begin{cases} \Delta \text{ if } d\in \mathcal{D}^{\prime},\\ \emptyset \text{ otherwise}.
    \end{cases}$\doublespacing
    \item Let $s, t$ be sequents whose $d$th component is $\Gamma^s_d$, or $\Gamma^t_d$, respectively. We write $s \times t$ (or: $st$) for the sequent whose $d$th component is $\Gamma^s_d\cup\Gamma^t_d$, for all $d\in \mathcal{D}$. We call $s, t$ \textit{sub-sequents} of any sequent $s \times t$.
    \end{itemize}
As a convention, we drop `\{', `\}' in the sequent environment $[\dots]$. We use Roman minuscules $s,t,r,\dots$ to represent arbitrary sub-sequents.

For example if $\mathcal{D}=\{1,2\}$, $[1:A,B]\times[2:B]=[1:A,B;2:B]$, whose set of all sub-sequents is $\{[\ ],[1:A],[1:B],[2:B],[1:A,B],[1:A;2:B],[1:B;2:B], [1:A,B;2:B]\}$.\footnote{Note that $[\ ]=[1:\emptyset;2:\emptyset], [1:A]=[1:A;2:\emptyset]$, etc.}
\end{definition}
\begin{definition}[rule, calculus]\

    \noindent A (sequent) \textit{rule} indicates permissible transformations from a multiset of sequents, called \textit{premises}, to a \textit{conclusion} (sequent), separated by a horizontal inference line. Formulae only appearing above the line are called \textit{active}, those only below the line \textit{principal} and those on both sides of the line \textit{parametric}. The multiset of all parametric formulae is called the \textit{context} of the rule.

    A (sequent) \textit{calculus} is a set of rules. 
\end{definition}
For example, the classical calculus \textbf{LK} \cite{gentzen1935auntersuchungen,gentzen1935buntersuchungen} is defined on 2-dimensional sequents whereas the Strong Kleene calculus \textbf{K3} \cite{multlog2024k3} is defined on 3-dimensional sequents. Table 1 shows the multiplicative sentential version of \textbf{LK} and Table 2 the multiplicative version of \textbf{K3}. The context restrictions ($^\ast, ^\dag, ^\ddag$) in Table 2 ensure compatibility with our minimal derivability relation. We will come back to this in \S 5.
\begin{table}
\fbox{
\parbox{\linewidth}{
Let $\mathcal{D}=\{1,2\}$ and let $i,j\in \mathcal{D}, i\neq j$. Multiplicative sentential \textbf{LK} consists of the \textit{structural rules}:
\begin{center}
    \begin{tabular}{c c c c}
       \begin{prooftree}
        \hypo{\emptyset}
        \infer1[\textsc{Id}]{[\mathcal{D}:A]}
    \end{prooftree}  & \begin{prooftree}
        \hypo{[i:A]s}
        \hypo{[j:A]t}
        \infer2[\textsc{Cut}]{st}
    \end{prooftree} &
    \begin{prooftree}
        \hypo{s}
        \infer1[$i$\textsc{W}]{[i:A]s}
    \end{prooftree} &
    \begin{prooftree}
        \hypo{[i:A,A]s}
        \infer1[$i$\textsc{C}]{[i:A]s}
    \end{prooftree}\\
         & 
    \end{tabular}   
\end{center}
and of the \textit{operational rules}:
\begin{center}
    \begin{tabular}{c c}
    \begin{prooftree}
        \hypo{[2:A]s}
        \infer1[1-]{[1:-A]s}
    \end{prooftree} 
    &
    \begin{prooftree}
        \hypo{[1:A]s}
        \infer1[2-]{[2:-A]s}
    \end{prooftree} 
    \\ \\
    \begin{prooftree}
        \hypo{[1:A,B]s}
        \infer1[$1\wedge$]{[1:A\wedge B]s}
    \end{prooftree} &
    \begin{prooftree}
        \hypo{[2:A]s}
        \hypo{[2:B]t}
        \infer2[$2\wedge$]{[2:A\wedge B]st}
    \end{prooftree} 
    \\ \\
    \begin{prooftree}
        \hypo{[1:A]s}
        \hypo{[1:B]t}
        \infer2[$1\vee$]{[1:A\vee B]st}
    \end{prooftree}
    &
    \begin{prooftree}
        \hypo{[2:A,B]s}
        \infer1[$2\vee$]{[2:A\vee B]s}
    \end{prooftree}
    \\ \\
    \begin{prooftree}
        \hypo{[2:A]s}
        \hypo{[1:B]t}
        \infer2[$1\supset$]{[1:A\supset B]st}
    \end{prooftree}
    &
    \begin{prooftree}
        \hypo{[1:A;2:B]s}
        \infer1[$2\supset$]{[2:A\supset B]s}
    \end{prooftree}
    \end{tabular}
\end{center}
}}
    \caption{multiplicative sentential \textbf{LK}}
\end{table}

\begin{table}[h]
\fbox{
\parbox{\linewidth}{
Let $\mathcal{D}=\{1,2,3\}$ and let $i,j,k\in \mathcal{D}$ be pairwise distinct. Multiplicative \textbf{K3} consists of the \textit{structural rules}:
\begin{center}
    \begin{tabular}{c c c c}
       \begin{prooftree}
        \hypo{\emptyset}
        \infer1[\textsc{Id}]{[\mathcal{D}:A]}
    \end{prooftree}  & \begin{prooftree}
        \hypo{[i:A]s}
        \hypo{[j:A]t}
        \infer2[\textsc{Cut}]{st}
    \end{prooftree} &
    \begin{prooftree}
        \hypo{s}
        \infer1[$i$\textsc{W}]{[i:A]s}
    \end{prooftree} &
    \begin{prooftree}
        \hypo{[i:A,A]s}
        \infer1[$i$\textsc{C}]{[i:A]s}
    \end{prooftree}\\
         & 
    \end{tabular}   
\end{center}
and of the \textit{operational rules}:
\begin{center}
\footnotesize
    \begin{tabular}{c c c}
    \begin{prooftree}
        \hypo{[3:A]s}
        \infer1[1-]{[1:-A]s}
    \end{prooftree} 
    &
    \begin{prooftree}
        \hypo{[2:A]s}
        \infer1[2-]{[2:-A]s}
    \end{prooftree} 
    &
    \begin{prooftree}
        \hypo{[1:A]s}
        \infer1[3-]{[3:-A]s}
    \end{prooftree} 
    \\ \\
    \begin{prooftree}
        \hypo{[1:A,B]s}
        \infer1[$1\wedge$]{[1:A\wedge B]s}
    \end{prooftree}
    &
    \begin{prooftree}
        \hypo{[2,3:A]s}
        \hypo{[2:A,B]t}
        \hypo{[2,3:B]u}
        \infer3[$2\wedge^\ast$]{[2:A\wedge B]stu}
    \end{prooftree} 
    &
    \begin{prooftree}
        \hypo{[3:A]s}
        \hypo{[3:B]t}
        \infer2[$3\wedge$]{[3:A\wedge B]st}
    \end{prooftree} 
    \\ \\
    \begin{prooftree}
        \hypo{[1:A]s}
        \hypo{[1:B]t}
        \infer2[$1\vee$]{[1:A\vee B]st}
    \end{prooftree}
    &
    \begin{prooftree}
        \hypo{[1,2:A]s}
        \hypo{[2:A,B]t}
        \hypo{[1,2:B]u}
        \infer3[$2\vee^\dag$]{[2:A\vee B]stu}
    \end{prooftree} 
    &
    \begin{prooftree}
        \hypo{[3:A,B]s}
        \infer1[$3\vee$]{[3:A\vee B]s}
    \end{prooftree}
    \\ \\
    \begin{prooftree}
        \hypo{[3:A]s}
        \hypo{[1:B]t}
        \infer2[$1\supset$]{[1:A\supset B]st}
    \end{prooftree}
    &
    \begin{prooftree}
        \hypo{[2,3:A]s}
        \hypo{[2:A,B]t}
        \hypo{[1,2:B]u}
        \infer3[$2\supset^\ddag$]{[2:A\supset B]stu}
    \end{prooftree} 
    &
    \begin{prooftree}
        \hypo{[1:A;3:B]s}
        \infer1[$3\supset$]{[3:A\supset B]s}
    \end{prooftree}
    \end{tabular}
    \end{center}
    \scriptsize
    $^\ast$$s=[1:\Gamma_1;2:\Gamma_2;3:\Gamma_3],t=[1:\Delta_1;2:\Delta_2;3:\Delta_3],u=[1:\Gamma_1, \Delta_1;2:\Sigma_2;3:\Sigma_3]$\\
    $^\dag$$s=[1:\Gamma_1;2:\Gamma_2;3:\Gamma_3],t=[1:\Delta_1;2:\Delta_2;3:\Delta_3],u=[1:\Sigma_1;2:\Sigma_2;3:\Gamma_3, \Delta_3]$\\
    \ddag$s=[1:\Gamma_1;2:\Gamma_2;3:\Gamma_3],t=[1:\Delta_1;2:\Delta_2;3:\Delta_3],u=[1:\Gamma_1,\Delta_1;2:\Sigma_2;3:\Gamma_3, \Delta_3]$
}}
    \caption{multiplicative \textbf{K3}}
\end{table}

 How should we read sequents? We typically think of a standard 2-dimensional sequent $\Gamma_1|\Gamma_2$ as expressing an (inner) derivability or validity relation: `Some formula in $\Gamma_2$ can be derived from all formulae in $\Gamma_1$' ($\bigwedge\limits_{p\in \Gamma_1} p\vdash \bigvee\limits_{q\in \Gamma_2}q$). We call this the \textit{derivability reading} of sequents. In the 2-dimensional setting, this is equivalent to the material conditional $\bigwedge\limits_{p\in \Gamma_1} p\supset \bigvee\limits_{q\in \Gamma_2}q$ and, hence, to  $\bigvee\limits_{p\in \Gamma_1} 
 - p\vee \bigvee\limits_{q\in \Gamma_2}q$, which gives rise to the \textit{truth-functional} reading: `One of the formulae in $\Gamma_1$ is false, or one of the formulae in $\Gamma_2$ is true'. 
 
 However, in a $>2$-dimensional setting, the equivalence between both readings generally breaks down. In response, MUltlog---in line with a well-developed tradition tracing back to \cite{rousseau1967sequents, rousseau1970sequents, schroter1955methoden}---adopts a (positive) truth-functional reading, in which each sequent dimension corresponds to a truth value. According to this reading, a 3-dimensional sequent $\Gamma_1|\Gamma_2|\Gamma_3$ expresses the following 3-valued validity statement, for truth values $1,2,3$: `At least one formula in $\Gamma_1$ is 1, or at least one formula in $\Gamma_2$ is 2, or at least one formula in $\Gamma_3$ is 3'.\footnote{For details and a literature overview and the difference between positive and negative truth-functional readings see, for example, \cite{wintein2016all, zach1993proof}. (2-dimensional) sequent calculi for $>2-$valued logics that are based on the derivability reading include e.g. hypersequent \cite{avron1991natural} or labelled \cite{cobreros2022higher} approaches.} 

As proof-theoretic semanticists, we might have qualms about relying on truth values to make sense of sequents, i.e. the proof-theoretic structures we use to give our semantics in. One can argue that a core motivating feature of P-tS is its ability to model meaning without recourse to notions of truth. We avoid such objections by adopting a \textit{multilateralist} sequent reading. This is an extension of Restall's \textit{bilateralist} reading \cite{restall2005multiple}, according to which the $2$-dimensional sequent $\Gamma_1|\Gamma_2$ corresponds to an incoherent (or: out-of bounds) state, i.e. `It is incoherent to both assert every formula in $\Gamma_1$ and deny everything in $\Gamma_2$'  (-($\bigwedge\limits_{p\in \Gamma_1} p\wedge \bigwedge\limits_{q\in \Gamma_2}-q)$). In lieu of truth values, we operate with the primitive speech acts of \textit{assertion} and \textit{denial}, which constrain (rational) discursive contributions of an agent.\footnote{Alternatively, one can work with \textit{acceptance} and \textit{rejection}, thought of as agents' rational cognitive states, or attitudes.} When adding further sequent dimensions, we, thus, expand our available speech acts accordingly. In the case of \textbf{K3}, we will add the speech act of \textit{suspending judgement}, yielding for the 3-dimensional sequent $\Gamma_1|\Gamma_2|\Gamma_3$: `It is incoherent to assert every formula in $\Gamma_1$, suspend judgement about everything in $\Gamma_2$ whilst denying everything in $\Gamma_3$'.

Recent advocacy for \textit{suspension of judgement} as an independent speech act can be found in \cite{friedman2013suspended, friedman2017suspend, friedman2013rational}. We can think of suspending judgement about something as \textit{refraining from asserting or denying} it \cite{ferrari2022varieties}.\footnote{See also \cite{incurvati2026plea} for a related notion of multilateralism.} Different multi-dimensional sequent calculi might require different speech acts just like different multi-valued logics might come with different truth values. For instance, $3-$dimensional sequents in a calculus for the Logic of Paradox \textbf{LP} (given by the same rules as \textbf{K3}) are best read using the third speech act of \textit{doubling judgement}. We can think of doubling judgment about something as \textit{asserting and denying} it---once again, conceived of as an independent and irreducible speech act. Hence, we obtain for the sequent $\Gamma_1|\Gamma_2|\Gamma_3$: `It is incoherent to assert every formula in $\Gamma_1$, doubling judgement about everything in $\Gamma_2$ whilst denying everything in $\Gamma_3$'.

This multilateral reading helps us make sense of sequents, our semantic base structure, and supplies some philosophical grounding for our approach. It allows us to view sequent rules (and derivations constructed from them) as constraining rational positions. For example, we can read the \textbf{K3} negation rules as `If denying $A$ is incoherent in the context $s$, so is asserting $-A$', `If suspending judgement about $A$ is incoherent in the context $s$, so is suspending judgement about $-A$', and `If asserting $A$ is incoherent in the context $s$, so is denying $-A$', respectively. It is important to note that this multilateralism is not part of our formal semantic theory but a device that helps make sense of the employed formalism (and, as such, replaceable by other readings). However, we can think it as one way of \textit{interpreting} our semantic theory alongside alternative options.

Finally, we employ a standard notion of (vertical) derivations \cite[e.g.][]{restall2023logical}:
\begin{definition}[vertical derivation, proof]\

    \noindent A (vertical) \textit{derivation} (tree) of sequent $\mathfrak{s}$ from the set of sequents $\mathfrak{S}$ in a calculus \textbf{Cal} is a non-downward branching partial order on a finite set of sequents, whose root is $\mathfrak{s}$ and whose leaves are $\emptyset$ or premisses, i.e. some $\mathfrak{s}^\prime \in \mathfrak{S}$, and every step is an application of some \textbf{Cal}-rule.
    
    If $\mathfrak{S}=\emptyset$, we call the derivation a (vertical) \textit{proof} (tree).
\end{definition}
The tree structure constructed from chained applications of rules represents a calculus's \textit{vertical} derivability relation, i.e. the derivability relation between its sequents, making vertical derivability calculus-dependent. In contrast, (vertically provable) sequents correspond to the notion of \textit{validity} encoded the calculus. We also call this the \textit{horizontal derivability relation} of \textbf{Cal}.
\begin{definition}[horizontal derivability, provability]\ 

    \noindent Some formula $A\in \Delta$ is horizontally \textit{derivable} from the multiset of formulae $\Gamma$ in a calculus \textbf{Cal} with the set of dimensions $\mathcal{D}$ iff the sequent $[\mathcal{D}^-:\Gamma;\mathcal{D}^+:\Delta]$ is vertically provable in \textbf{Cal}, where $\mathcal{D}^-$ is the set of \textit{undesignated} dimensions in \textbf{Cal}, $\mathcal{D}^+$ is the set of \textit{designated} dimensions in \textbf{Cal} and $\mathcal{D}=\mathcal{D}^- \cup \mathcal{D}^+$. 

    If $\Gamma=\emptyset$, we say that some $A\in \Delta$ is horizontally \textit{provable} in \textbf{Cal}.
\end{definition}
The designated dimensions of an $n$-dimensional calculus generally correspond to the designated truth values of its corresponding model theory.

\begin{example}\

    \noindent
    \begin{enumerate}
        \item \textbf{LK} has the designated value $\mathcal{D}^+=\{2\}$. Hence, $A$ is horizontally derivable from $\Gamma$ in \textbf{LK} iff $[1:\Gamma;2:A]$ is vertically provable in \textbf{LK}. Moreover, $A$ is \textbf{LK}-provable iff $[2:A]$ is vertically provable in \textbf{LK}. Using our multilateralism, we can understand this to mean that denying $A$ is incoherent---$A$ is \textit{undeniable}.
        \item \textbf{K3} has the designated value $\mathcal{D}^+=\{3\}$. Hence, some $A\in \Delta$ is horizontally derivable from $\Gamma$ in \textbf{K3} iff $[1,2:\Gamma;3:\Delta]$ is vertically provable in \textbf{K3}. Moreover, some $A \in \Delta$ is \textbf{K3}-provable iff $[3:\Delta]$ is vertically provable in \textbf{K3}. 
        \item \textbf{LP} has the designated values $\mathcal{D}^+=\{2,3\}$. Hence, $A$ or $B$ is horizontally derivable from $\Gamma$ in \textbf{LP} iff $[1:\Gamma;2:A,B;3:A,B]$ is vertically provable in \textbf{LP}. Moreover, $A$ or $B$ is horizontally \textbf{LP}-provable iff $[2,3:A,B]$ is vertically provable in \textbf{LP}. 
    \end{enumerate}
\end{example}
\section{Methodology: Inference-Behaviour Semantics}
Given our formal apparatus, let us now briefly introduce our semantic toolbox: inference-behaviour semantics (I-bS). We will limit ourselves to presenting the key ideas; technical details can be found in \cite{nagler2026inference} and philosophical motivations in \cite{nagler2026measuring}.

I-bS---and P-tS in general---is an \textit{inferentialist} approach to semantics, according to which the meaning of a connective is fixed by its use in reasoning \cite{brandom1994making, incurvati2023reasoning}. It is a \textit{formal} approach insofar as it views logical proofs as representing such reasoning. What makes I-bS distinct from other forms of P-tS is the kind of connective use it investigates. \textit{Local} approaches such as \cite{ferrari2021proof, hjortland2013logical, restall2014pluralism} aim to find the meaning of a connective in constituents of their operational rules alone. However, by merely considering rules but not their application instances in proofs they fail to consider features of a connective's meaning that only become visible on the level of derivability such as structural co-determination effects \cite{dicher2016proof, nagler2026measuring}. In contrast, \textit{global} approaches such as proof-theoretic validity \cite{prawitz1971ideas,prawitz1974idea, prawitz1973towards}, base-extension semantics \cite{sandqvist2015base} or earlier substructural approaches \cite{villalonga2020substructural} investigate connective use in full logical calculi. As a result, connective use becomes inseparably interwoven with the notion of derivability of a specific logic. Hence, one cannot compare connective meanings across different logics, limiting the explanatory power of such approaches and rendering them unsuitable, for instance, for modelling meaning-invariance relations amongst logics such as in frameworks for logical pluralism \cite[e.g.][]{ferrari2021proof, restall2014pluralism}.

I-bS aims to strike a balance between such local and global approaches. In \cite{nagler2026measuring}, we dubbed such an approach \textit{regional}, drawing inspiration from \cite{cobreros2011supervaluationism}. Whilst I-bS investigates connective use in actual proofs---not merely operational rules---it limits itself to proofs that are \textit{characteristic} for the use of the connective. Specifically, I-bS focusses on the use that occurs in proofs of the connective's \textit{definability} relative to a substructural \textit{minimal derivability relation}. The aim is to make the connective's use visible whilst minimising confounding structural input that would result from using a full-strength derivability relation instead. Whilst there is no unique \textit{minimal derivability relation}, its choice is limited by pragmatic desiderata such as the compatibility with the derivability relations of the logics under investigation or  comparability with previous results.

Following the reasoning in \cite{nagler2026measuring}, we will focus on the derivability relation consisting of rules for (context-differentiating) \textsc{Cut} and \textsc{Id}. Whilst our 3-dimensional setting allows for various formulations of \textsc{Cut} and \textsc{Id} rules, including 2-dimensional \textsc{Id$^2$} and 3-dimensional \textsc{Cut$^3$},
\[
\begin{prooftree}
    \hypo{\emptyset}
    \infer1[\textsc{Id}$^2$]{[i,j:A]}
\end{prooftree}\qquad
\begin{prooftree}
    \hypo{[1:A]s}
    \hypo{[2:A]t}
    \hypo{[3:A]u}
    \infer3[\textsc{Cut}$^3$]{stu}
\end{prooftree},
\]
we will stick to the standard combination of 3-dimensional \textsc{Id$^3$} and 2-dimensional \textsc{Cut$^2$}: 
\[
\begin{prooftree}
        \hypo{\emptyset}
        \infer1[\textsc{Id$^3$}]{[\mathcal{D}:A]}
    \end{prooftree}  \qquad\begin{prooftree}
        \hypo{[i:A]s}
        \hypo{[j:A]t}
        \infer2[\textsc{Cut$^2$}]{st}
    \end{prooftree},
\]
for $\mathcal{D}=\{1,2,3\}, i,j\in \mathcal{D}, i\neq j$. Our reasons for this include that \textsc{Id$^3$} and \textsc{Cut$^2$} are the default versions, including for use in MUltlog \cite{multlog2024lp, multlog2024k3}, and the different $3$-dimensional formulations are not generally equivalent (e.g. $\textsc{Id}^3+\textsc{Cut}^2\vdash \textsc{Cut}^3$ and $\textsc{Id}^2+\textsc{Cut}^3\vdash \textsc{Id}^3$, but $ \textsc{Id}^3+\textsc{Cut}^2\not{\vdash} \textsc{Id}^2$ and $\textsc{Id}^2+\textsc{Cut}^3\not{\vdash} \textsc{Cut}^2$).\footnote{In the case of mixed definitions, $\textsc{Id}^3+\textsc{Cut}^3\not{\vdash} \textsc{Cut}^2$ and $\textsc{Id}^3+\textsc{Cut}^3\not{\vdash} \textsc{Id}^2$. Notably, $\textsc{Id}^2+ \textsc{Cut}^2$ trivialises any $>2$-dimensional sequent calculus in the presence of contraction.}

For our notion of \textit{definability}, we use a Belnap-style definition \cite{belnap1962tonk}:
\begin{definition}[definability]\

    \noindent
    Let $\mathfrak{R}_\#$ be the set of operational rules for the connective \#, and let \textbf{Cal} be a sequent calculus. $\mathfrak{R}_\#$ and \textbf{Cal} are defined for all and only the dimensions in $\mathcal{D}$.

    \# is \textit{definable} in \textbf{Cal} iff
    \begin{enumerate}
        \item \# \textit{conservatively extends} \textbf{Cal}, i.e. all sequents provable in $\textbf{Cal}\cup \mathfrak{R}_\#$ contain `\#', unless they are also provable in \textbf{Cal} (\textit{conservativity}), and
        \item \#-formulae and $\natural$-formulae are (horizontally) \textit{inter-derivable} in $\textbf{Cal}\cup \mathfrak{R}_\#\cup\mathfrak{R}_\natural$, where $\mathfrak{R}_\#=\mathfrak{R}_\natural$ modulo substituting `$\natural$' for `\#', i.e. for any \#-formula $\texttt{\#}$ and each $\natural$-formula $\na$, $[\mathcal{D}^-:\texttt{\#};\mathcal{D}^+:\na]$ is provable in $\textbf{Cal}\cup \mathfrak{R}_\#\cup\mathfrak{R}_\natural$  (\textit{uniqueness}).
    \end{enumerate}
\end{definition}
Intuitively, \textit{conservativity} ensures that \# has an inferential role---and, hence, semantic content---which is distinct from that of other connectives. Meanwhile, \textit{uniqueness} certifies that this inferential role is \textit{unique} up to isomorphism.

When introducing I-bS, we defined uniqueness in terms of vertical rather than horizontal interderivability \cite{nagler2026measuring}. In \cite{nagler2026inference}, we weighed the two options against each other, coming to no conclusive decision in favour of either. In this paper, we select the \textit{horizontal} option for two reasons: firstly, previous work using vertical inter-derivability focussed on 2-dimensional calculi. In the 2-dimensional case, vertical inter-derivability---as fixed by a calculus---and horizontal inter-derivability stand in a 1:1 relation since both are binary relations.\footnote{Technically, there are \textit{four} horizontal derivability relations for each 2-dimensional calculus. However, two of these are trivial and the other two are identical up to converse. For instance, for the calculus of \textbf{LK}, we also get calculi with $\mathcal{D}^+=\{1,2\}$ and $\mathcal{D}^+=\emptyset$---the trivial options for which all sequents are valid or invalid, respectively---and the one with $\mathcal{D}^+=\{1\}$ where each horizontal validity $\Gamma \vdash \Delta$ corresponds to the \textbf{LK}-validity $\Delta \vdash \Gamma$ (the latter for $\mathcal{D}^+=\{2\}$).}  However, this correspondence breaks down in the $> 2$-dimensional case: one vertical derivability relation can correspond to multiple horizontal derivability relations. For instance, the vertical derivability relation defined by the rules in Table 2 corresponds to the horizontal derivability relations of both \textbf{K3} and \textbf{LP} and is only distinguished by the choice of designated dimension (cf. Example 1, \S2). Whilst the vertical derivability relation is still binary, sequents are \textit{ternary} and only collapsed into a binary horizontal derivability relation by fiat.

Secondly and relatedly, we typically prove vertical inter-derivability in terms of horizontal inter-derivability. In the $2$-dimensional case, this means first establishing $[i:\texttt{\#};j:\na]$ ($i\neq j$) and then using \textsc{Cut} to derive $[k:\texttt{\#}]s$ from $[l:\texttt{\na}]s$ ($k\neq l$), and vice versa. In the $\geq 2$-dimensional case, this strategy is not available in our minimal derivability relation alone but requires further structural resources such as the presence of \textsc{Contraction}. Presupposing these counteracts the minimalist project of I-bS \textit{ex ante}.

We now have the context within which we measure the use and, hence, meaning of connectives: Belnap-style definability proofs relative to a minimal derivability relation of $\textsc{Id}^3, \textsc{Cut}^2$.\footnote{For the remainder of this paper, we will drop the superscripts.} We now introduce the centrepiece of I-bS: our metric for connective use, \textit{inference behaviour}.

\begin{definition}[inference behaviour $\texttt{ib}$]\

    \noindent 
    Let \textbf{Cal} be a calculus with the set of dimensions $\mathcal{D}$. Let $\#$ be a fixed connective and let $\mathfrak{p}$ be a list of non-empty tree nodes (\textit{vertical locations} $V$) forming a \textbf{Cal}-proof tree. Let $A^{\#}$ be a $\#$-formula, or the active formula of some (later) application of a $\#$-rule.

    The \textit{inference behaviour} \texttt{ib} of $\#$ at a vertical location $V$ is the function
    \[
    \texttt{ib}(\#,V) := \langle H(\#,V), R(\#,V)\rangle\text{,  such that:}
    \]
    \begin{itemize}
        \item the \textit{horizontal-location} function $H$ maps each pair $\langle\#,V\rangle$ to a multiset $\mathcal{D}'$ over $\mathcal{D}$, where each $d \in \mathcal{D}'$ represents the dimension in which $A^{\#}$ occurs, for all $A^{\#}$ occurring in the sequent at $V$;
        \item the \textit{rule-application} function $R$ maps each pair $\langle \#,V\rangle$ to an element of \textbf{Cal}$\cup\{\texttt{none}\}$, and $R(\#,V)$ is the rule applied at the inference line immediately above $V$, whenever some $A^{\#}$ is a principal formula of $R(\#,V)$, and \texttt{none} otherwise.
    \end{itemize}
    The \textit{inference behaviour} $\texttt{ib}^\#_{\mathfrak{p}}$ of $\#$ in the proof $\mathfrak{p}$ is the list of values $\texttt{ib}(\#,V)$ for each $V \in \mathfrak{p}$, namely $\texttt{ib}^\#_{\mathfrak{p}}= \big\langle \langle H(\#,V), R(\#,V)\rangle \big\rangle_{V \in \mathfrak{p}}$.

The \textit{inference-behaviour profile} of $\#$ in $\mathfrak{p}$ is
  $\texttt{ibp}^\#_{\mathfrak{p}}
  = \{ H(\#,V) \mid V \in \mathfrak{p} \}
    \cup
    \{ R(\#,V) \mid V \in \mathfrak{p} \}$.
\end{definition}
Intuitively, $\mathtt{ib}^\#_\mathfrak{p}$ tells us \textit{how} \# is used ($H$) and \textit{why} it is used ($R$) \textit{whenever} it is used ($V$) in $\mathfrak{p}$. $V$ book-keeps the proof steps in which \# is used. $H$ keeps track of the dimension(s) in which the \#-usage occurs ($\emptyset$ if there is no \# usage). Using our trilateral sequent reading, we can think of this as the speech act under which \# is considered. $R$ records the rule used to justify the \#-use (\texttt{none} if there is no \# use).

\begin{example}\

     \noindent Let $\mathfrak{p}$ correspond to the following \textbf{K3} proof:
     \[
     \begin{prooftree}
         \hypo{\emptyset}
         \infer1[\textsc{Id}]{[1:A;2:A;3:A]^2}
         \hypo{\emptyset}
         \infer1[\textsc{Id}]{[1:A;2:A;3:A]^1}
         \infer1[$2W$]{[1:A;2:A,B;3:A]^3}
         \hypo{\emptyset}
         \infer1[\textsc{Id}]{[1:B;2:B;3:B]^4}
         \infer3[$2\supset$]{[1:A,A;2:A\supset B;3:A,B]^5}
         \infer1[$1C$]{[1:A;2:A\supset B;3:A,B]^6}
         \infer1[$2\vee $]{[1:A;2:A\supset B;3:A\vee B]^7}
     \end{prooftree}
     \]
     We order vertical locations starting at the top left and moving in English reading direction left-to-right, top-to-bottom (cf. the displayed superscripts). We use `$\{\dots\}$' for sets, `$[\dots]$' for multisets, and `$\langle \dots \rangle$' for lists.
     
     By tracking the occurrences of $\supset$-formulae and of their active formulae across the 7 proof steps, we get (from location 1 to 7): 
     \[
     \mathtt{ib}^\supset_\mathfrak{p}=\langle\langle[2], \textsc{Id}\rangle,\langle[2,3], \textsc{Id}\rangle,\langle[2,2],2W\rangle,\langle[1,2],\textsc{Id}\rangle,\langle[2], 2\supset \rangle,\langle[2],\texttt{none}\rangle,\langle[2], \texttt{none}\rangle\rangle.
     \]
     $\mathtt{ibp}^\supset_\mathfrak{p}=\{[2],[2,3],[2,2],[1,2],\textsc{Id}, 2W,2\supset ,\texttt{none}\}$ keeps track of the types of horizontal locations and rules recorded in $\mathtt{ib}^\supset_\mathfrak{p}$.
     
     One can visualise $\mathtt{ib}^\supset_\mathfrak{p}$ intuitively by boxing the tracked formulae and rules:
    \[
     \begin{prooftree}
        \hypo{\emptyset}
         \infer1[\boxed{\textsc{Id}}]{[1:A;2:\boxed{A};3:\boxed{A}]^2}
         \hypo{\emptyset}
         \infer1[\boxed{\textsc{Id}}]{[1:A;2:\boxed{A};3:A]^1}
         \infer1[$\boxed{2W}$]{[1:A;2:\boxed{A,B};3:A]^3}
         \hypo{\emptyset}
         \infer1[\boxed{\textsc{Id}}]{[1:\boxed{B};2:\boxed{B};3:B]^4}
         \infer3[$\boxed{2\supset}$]{[1:A,A;2:\boxed{A\supset B};3:A,B]^5}
         \infer1[$1C$]{[1:A;2:\boxed{A\supset B};3:A,B]^6}
         \infer1[$2\vee$]{[1:A;2:\boxed{A\supset B};3:A\vee B]^7}
     \end{prooftree}
     \]
     If we do the same for $\mathtt{ib}^\vee_\mathfrak{p}$, we get 
          \[
     \begin{prooftree}
         \hypo{\emptyset}
         \infer1[\textsc{Id}]{[1:A;2:A;3:A]^2}
         \hypo{\emptyset}
         \infer1[\boxed{\textsc{Id}}]{[1:A;2:A;3:\boxed{A}]^1}
         \infer1[$2W$]{[1:A;2:A,B;3:\boxed{A}]^3}
         \hypo{\emptyset}
         \infer1[\boxed{\textsc{Id}}]{[1:B;2:B;3:\boxed{B}]^4}
         \infer3[$2\supset$]{[1:A,A;2:A\supset B;3:\boxed{A,B}]^5}
         \infer1[$1C$]{[1:A;2:A\supset B;3:\boxed{A,B}]^6}
         \infer1[$\boxed{2\vee}$]{[1:A;2:A\supset B;3:\boxed{A\vee B}]^7}
     \end{prooftree}
     \]
     and, hence, $\mathtt{ib}^\vee_\mathfrak{p}=\langle\langle [3], \textsc{Id}\rangle,\langle\emptyset, \texttt{none}\rangle, \langle[3],\texttt{none}\rangle,\langle[3],\textsc{Id}\rangle,\langle[3,3],\texttt{none}\rangle,\langle[3,3],\texttt{none}\rangle,$\linebreak $\langle[3], 2\vee\rangle\rangle$, and $\mathtt{ibp}^\vee_\mathfrak{p}=\{ \emptyset, [3],[3,3], \textsc{Id}, 2\vee, \texttt{none}\}$.
\end{example}
Putting together our three building blocks, we can now give the I-bS definition for meaning identity:
\begin{definition}[identity of connective meanings]\

    \noindent
    Two connectives have the same \textit{meaning} iff they have the same \textit{inference behaviour} in the proof of their \textit{definability} relative to a fixed \textit{minimal derivability relation} (up to substitution of connective symbols).
\end{definition}
In our concrete case, we can give the meaning of a connective by proving its \textit{conservativity} and \textit{uniqueness} over $\{\textsc{Id}, \textsc{Cut}\}$ and measuring its \textit{inference behaviour}. We can give this meaning in the form of those substructural operational rules that accommodate all and only the structural properties (relational and set-theoretic properties of the derivability relation $\vdash$) tracked in its \textit{inference behaviour} in these proofs. These rules are called the \textit{semantic clause} of the connective.\footnote{See \cite{nagler2026inference} for more details on semantic clauses in I-bS.}

For example, we can see the semantic clauses for all connectives of sentential \textbf{LK} in Table 3 as proven in \cite{nagler2026inference, nagler2026measuring}.\footnote{Despite the calculus-like presentation, one should think of the semantic clause for each distinct connective as being \textit{partitioned off} from the others and not interacting \cite[see][]{nagler2026inference}.}
\begin{table}[h]
\fbox{
\parbox{\linewidth}{
    \centering
    \footnotesize
    \begin{center}
    \begin{tabular}{c c}
    \begin{prooftree}
        \hypo{[1:C;2:A]}
        \infer1[1$\neg$]{[1:\neg A, C]}
    \end{prooftree} 
    &
    \begin{prooftree}
        \hypo{[1:A,C]}
        \infer1[2$\neg$]{[1:C;2:\neg A]}
    \end{prooftree} 
    \\ \\
        \begin{prooftree}
        \hypo{[2:A, C]}
        \infer1[1$\sim$]{[1:\ \sim A; 2: C]}
    \end{prooftree} 
    &
    \begin{prooftree}
        \hypo{[1:A;2:C]}
        \infer1[2$\sim$]{[2:\ \sim A, C]}
    \end{prooftree} 
    \\ \\
    \begin{prooftree}
        \hypo{[1:A,B; 2:C]}
        \infer1[$1\otimes$]{[1:A\otimes B; 2:C]}
    \end{prooftree} &
    \begin{prooftree}
        \hypo{[1:C;2:A]}
        \hypo{[1:D;2:B]}
        \infer2[$2\otimes$]{[1: C,D; 2:A\otimes B]}
    \end{prooftree} 
    \\ \\
    \begin{prooftree}
        \hypo{[1:A; 2:C]}
        \infer1[$1\sqcap$1]{[1:A\sqcap B; 2:C]}
    \end{prooftree} \ 
    \begin{prooftree}
        \hypo{[1:B; 2:C]}
        \infer1[$1\sqcap$2]{[1:A\sqcap B; 2:C]}
    \end{prooftree}&
    \begin{prooftree}
        \hypo{[1:C;2:A]}
        \hypo{[1:C;2:B]}
        \infer2[$2\sqcap$]{[1: C; 2:A\sqcap B]}
    \end{prooftree} 
    \\ \\
    \begin{prooftree}
        \hypo{[1:A;2:C]}
        \hypo{[1:B;2:D]}
        \infer2[$1\oplus$]{[1:A\oplus B;2:C,D]}
    \end{prooftree}
    &
    \begin{prooftree}
        \hypo{[1:C;2:A,B]}
        \infer1[$2\oplus$]{[1:C; 2:A\oplus B]}
    \end{prooftree}
    \\ \\
    \begin{prooftree}
        \hypo{[1:A;2:C]}
        \hypo{[1:B;2:C]}
        \infer2[$1\sqcup$]{[1:A\sqcup B;2:C]}
    \end{prooftree}
    &
    \begin{prooftree}
        \hypo{[1:C;2:A]}
        \infer1[$2\sqcup$1]{[1:C; 2:A\sqcup B]}
    \end{prooftree} \
    \begin{prooftree}
        \hypo{[1:C;2:B]}
        \infer1[$2\sqcup$2]{[1:C; 2:A\sqcup B]}
    \end{prooftree}
    \\ \\
    \begin{prooftree}
        \hypo{[1:C;2:A]}
        \hypo{[1:B;2:D]}
        \infer2[1$\rightarrow$]{[1:A\rightarrow B, C;2:D]}
    \end{prooftree}
    &
    \begin{prooftree}
        \hypo{[1:A,C;2:B]}
        \infer1[2$\rightarrow$]{[1:C;2:A\rightarrow B]}
    \end{prooftree}
    \\ \\
       \begin{prooftree}
        \hypo{[1:C;2:A]}
        \hypo{[1:B,C]}
        \infer2[1$\rightsquigarrow$]{[1:A\rightsquigarrow B, C]}
    \end{prooftree}
    &
    \begin{prooftree}
        \hypo{[1:A,C]}
        \infer1[2$\rightsquigarrow$1]{[1:C;2:A\rightsquigarrow B]}
    \end{prooftree} \
        \begin{prooftree}
        \hypo{[1:B,C]}
        \infer1[2$\rightsquigarrow$2]{[1:C;2:A\rightsquigarrow B]}
    \end{prooftree}
    \end{tabular}
\end{center}}}
    \caption{semantic clauses, sentential \textbf{LK}}
    \label{tab:placeholder}
\end{table}
Notably, there are (at least) twice as many semantic clauses and, thus, meaningful connectives as one typically gives operational rules for. As readers are likely aware, the rules of \textbf{LK} can be presented in two equivalent versions: \textit{multiplicative} and \textit{additive}. We presented the former in Table 1; the latter can be obtained from the former by using the same context all sequents in 2-premiss rules, or by splitting 1-premiss rules with 2 different active formulae into two identical rules, except that they only contain one of the active formulae. For instance, the additive version for $\supset$ is 
\[
     \begin{prooftree}
        \hypo{[2:A]s}
        \hypo{[1:B]s}
        \infer2[1$\supset^{add}$]{[1:A\supset^{add} B]s}
    \end{prooftree}
    \qquad
    \begin{prooftree}
        \hypo{[1:A]s}
        \infer1[2$\supset^{add}$1]{[2:A\supset^{add} B]s}
    \end{prooftree} \
        \begin{prooftree}
        \hypo{[1:B]s}
        \infer1[2$\supset^{add}$2]{[2:A\supset^{add} B]s}
    \end{prooftree}
\]
Crucially, multiplicative and additive versions of the same \textbf{LK}-connective have different semantic clauses---and, thus, \textit{meanings}---in I-bS. Subsequently, I-bS explains their equivalence in classical logic \textbf{LK} (`confounding') not as a semantic property of the connective but as a feature of the (fully structural) derivability relation of \textbf{LK}. In contrast, only the (sub-)structural features encoded in the semantic clauses (e.g. the cardinality of each sequent component) are part of the meaning of the connective. One can think of the I-bS semantic clauses as the semantic core of the multiplicative (cf. $\otimes, \oplus,\rightarrow$) and additive (cf. $\sqcap, \sqcup, \rightsquigarrow$) versions of the \textbf{LK} connective rules. An equivalent split can be found for negation ($\neg$ vs. $\sim$).\footnote{For details, see \cite{nagler2026inference}.}

The features of this semantic analysis become especially salient when comparing connective meanings across different calculi. In \cite{nagler2026inference, nagler2026measuring}, we have provided I-bS for intuitionistic, dual-intuitionistic, lattice, minimal, relevant and linear logics, yielding, for example, the following semantic insights: `disjunction' in \textbf{LK} corresponds to two separate meaningful connectives, $\oplus$ and $\sqcup$, conflated by the derivability relation of \textbf{LK}. The same connectives can be found in additive-multiplicative linear logic, albeit differentiated on a calculus-level here. Hence, `or' means the same in linear and classical logic; the use differences in linear and classical reasoning, respectively, are due to the derivability relation, not the meaning of `$\oplus$' or `$\sqcup$'. In contrast, when giving I-bS for intuitionistic \textbf{LJ}, we only obtain one semantic clause for $\sqcup$. Hence, disjunction in intuitionistic and that in classical logic have different (albeit compatible) meanings.

\begin{table}[h]
\fbox{
\parbox{\linewidth}{
\begin{center}
\small
    \begin{tabular}{c c c}
    \begin{prooftree}
        \hypo{[3:A]s}
        \infer1[1-]{[1:-A]s}
    \end{prooftree} 
    &
    \begin{prooftree}
        \hypo{[2:A]s}
        \infer1[2-]{[2:-A]s}
    \end{prooftree} 
    &
    \begin{prooftree}
        \hypo{[1:A]s}
        \infer1[3-]{[3:-A]s}
    \end{prooftree} 
    \\ \\
    \begin{prooftree}
        \hypo{[1:A]s}
        \infer1[$1\wedge$1]{[1:A\wedge B]s}
    \end{prooftree}\ \begin{prooftree}
        \hypo{[1:B]s}
        \infer1[$1\wedge$2]{[1:A\wedge B]s}
    \end{prooftree}
    &
    \multicolumn{2}{c}{
    \begin{prooftree}
        \hypo{[2,3:A]s^\prime}
        \hypo{[2:A]s}
        \hypo{[2,3:B]s^\prime}
        \infer3[$2\wedge1^\ast$]{[2:A\wedge B]s}
    \end{prooftree}}
    \\ \\
    \multicolumn{2}{c}{
    \begin{prooftree}
        \hypo{[2,3:A]s^\prime}
        \hypo{[2:B]s}
        \hypo{[2,3:B]s^\prime}
        \infer3[$2\wedge2^\ast$]{[2:A\wedge B]s}
    \end{prooftree}}
    &
    \begin{prooftree}
        \hypo{[3:A]s}
        \hypo{[3:B]s}
        \infer2[$3\wedge$]{[3:A\wedge B]s}
    \end{prooftree} 
    \\ \\
    \begin{prooftree}
        \hypo{[1:A]s}
        \hypo{[1:B]s}
        \infer2[$1\vee$]{[1:A\vee B]s}
    \end{prooftree}
    &
    \multicolumn{2}{c}{
    \begin{prooftree}
        \hypo{[1,2:A]s^\prime}
        \hypo{[2:A]s}
        \hypo{[1,2:B]s^\prime}
        \infer3[$2\vee1^\dag$]{[2:A\vee B]s}
    \end{prooftree}}
    \\ \\
    \multicolumn{2}{c}{
    \begin{prooftree}
        \hypo{[1,2:A]s^\prime}
        \hypo{[2:B]s}
        \hypo{[1,2:B]s^\prime}
        \infer3[$2\vee2^\dag$]{[2:A\vee B]s}
    \end{prooftree}}
    &
    \begin{prooftree}
        \hypo{[3:A]s}
        \infer1[$3\vee$1]{[3:A\vee B]s}
    \end{prooftree}\ 
    \begin{prooftree}
        \hypo{[3:B]s}
        \infer1[$3\vee$2]{[3:A\vee B]s}
    \end{prooftree}
    \\ \\
    \begin{prooftree}
        \hypo{[3:A]s}
        \hypo{[1:B]s}
        \infer2[$1\supset$]{[1:A\supset B]s}
    \end{prooftree}
    &
    \multicolumn{2}{c}{
    \begin{prooftree}
        \hypo{[2,3:A]s^\prime}
        \hypo{[2:A]s}
        \hypo{[1,2:B]s^\prime}
        \infer3[$2\supset1^\ddag$]{[2:A\supset B]s}
    \end{prooftree}}
    \\ \\
    \multicolumn{2}{c}{
    \begin{prooftree}
        \hypo{[2,3:A]s^\prime}
        \hypo{[2:B]s}
        \hypo{[1,2:B]s^\prime}
        \infer3[$2\supset2^\ddag$]{[2:A\supset B]s}
    \end{prooftree}}
    &
    \begin{prooftree}
        \hypo{[1:A]s}
        \infer1[$3\supset$1]{[3:A\supset B]s}
    \end{prooftree}\
    \begin{prooftree}
        \hypo{[3:B]s}
        \infer1[$3\supset$2]{[3:A\supset B]s}
    \end{prooftree}
    \end{tabular}
\end{center}
\scriptsize
$^\ast s=[1:\Gamma_1;2:\Gamma_2;3:\Gamma_3],s^\prime=[1:\Gamma_1;2:\Gamma_2]$\\
$^\dag s=[1:\Gamma_1;2:\Gamma_2;3:\Gamma_3],s^\prime=[2:\Gamma_2;3:\Gamma_3]$\\
$^\ddag s=[1:\Gamma_1;2:\Gamma_2;3:\Gamma_3],s^\prime=[2:\Gamma_2]$
}}
    \caption{additive version, operational \textbf{K3} rules}
\end{table}

\textbf{K3} can be defined in both a multiplicative (Table 2) and an additive guise (Table 4).\footnote{Again, we require context restrictions ($^\ast, ^\dag, ^\ddag$) for compatibility with our minimal derivability relation. We will come back to this in \S 5.} Their equivalence is equally trivial in the presence of \textsc{$i$C} and \textsc{$i$W} rules. Just as in case of \textbf{LK}, we will see that I-bS classifies each version of each connective as semantically distinct, considering their equivalence a feature of \textbf{K3}'s derivability relation rather than its semantics. However, this analogy between \textbf{K3} and \textbf{LK} naturally breaks down when considering the added third dimension.

Here lies the main challenge of this paper: if we want to automatedly generate I-bS for any $n$-dimensional logic, we need to find a way to preserve its core motivating feature of inter-calculus semantic comparisons across different dimensions. I-bS has only ever been applied to 2-dimensional sequent calculi, generating 2-dimensional semantic clauses. However, an $n$-dimensional logic naturally has $n$-dimensional semantic clauses. Semantic clauses of different dimensionality cannot possibly yield the same minimal inference behaviour, the I-bS criterion for meaning identity. Typically, the identity of minimal inference behaviour would result in the identity of semantic clauses. Hence, the goal is to find a comparison device that is weaker than semantic clause identity without collapsing into total incomparability across dimensions. 

Our solution is the new concept of a \textit{meaning-extension}. In essence, a higher-dimensional connective extends the meaning of a lower-dimensional connective if the semantic clause of the former contains the semantic clause of the latter as a sub-sequent structure. More precisely:
\begin{definition}[meaning-extension]\

    \noindent Let \# be an $m$-dimensional connective, and let $\natural$ be an $n$-dimensional connective, for some finite $m,n\in \mathbb{N},m>n$. Hence, the operational rules of \# are specified for the set of dimensions $\mathcal{D}_m=\{1,2,\dots, m\}$, and those of $\natural$ for the set of dimensions $\mathcal{D}_n=\{1,2,\dots, n\}$.

    We say that \# \textit{extends the meaning} of $\natural$ iff there is a subset of dimensions $\mathcal{D}_o=\{j,k,\dots,o\}$  with $\mathcal{D}_o\subset \mathcal{D}_m$ and $|\mathcal{D}_o|=|\mathcal{D}_n|$, and there is an $n$-dimensional connective $\#^\prime$ such that
    \begin{enumerate}
        \item for each dimension $i \in \mathcal{D}_o$, the operational rule of the semantic clause for $\#^\prime$, i$\#^\prime$, is identical to the operational rule of the semantic clause for $\#$, i$\#$, except that each sequent $[\Gamma_j|\Gamma_k|\dots|\Gamma_o]$ in i$\#^\prime$ is an $n$-dimensional sub-sequent of the corresponding $m$-dimensional sequent $[\Gamma^\prime_1|\Gamma^\prime_2|\dots|\Gamma^\prime_m]$ in i$\#$ where $\Gamma_d=\Gamma_d^\prime$, for each $d\in \mathcal{D}_o$, and
        \item $\#^\prime$ has the same meaning as $\natural$.
    \end{enumerate}
\end{definition}
We can think of meaning-extensions as applying Belnap-style \textit{conservativity} at the level of dimensions: a higher-dimensional connective extends the meaning of its lower-dimensional counterpart if after adding the higher dimension(s)---i.e. extending existing rules with the new dimension(s) and adding connective rules for the added dimensions---no new connective occurrences are provable, except in the added dimension(s). Put in terms of our multilateralist sequent reading, adding new speech acts (e.g. suspending/doubling judgement) to the discursive space does not change the existing speech acts (e.g. asserting/denying) we perform with our connectives but only adds a new discursive context to consider.

We will find that the \textbf{LK}-semantic clauses extend the meaning of their \textbf{K3} counterparts as the former are contained in the first and third dimension of the latter. To be able to see this result, we first give I-bS for \textbf{K3}.
\section{Data Generation: Uniqueness and Conservativity}
We first prove uniqueness and conservativity (Definition 6) for the semantic clauses of additive and multiplicative \textbf{K3} over our minimal derivability relation, $\{\textsc{Id}, \textsc{Cut}\}$. We then measure the inference behaviour in these definability proofs to obtain the semantic clauses of \textbf{K3}.

\begin{theorem}[uniqueness]\

    \noindent For each multiplicative and additive \textbf{K3} connective \# with operational rules $\mathfrak{R}_\#$, \#-formulae and $\#^{\prime}$-formulae are \textit{inter-derivable} in $\{\textsc{Id}, \textsc{Cut}\}\cup \mathfrak{R}_\#\cup\mathfrak{R}_{\#^{\prime}}$, where $\mathfrak{R}_\#=\mathfrak{R}_{\#^{\prime}}$ modulo substituting `$\#^{\prime}$' for `\#'.
\end{theorem}
\begin{proof}\

\noindent 
\[
    \begin{prooftree}
    \hypo{[1,2,3:A]}
    \infer1[$1-$]{[1:-A, A;2:A]}  
    \infer1[$3-^{\prime}$]{[1:-A;2:A; 3:-^{\prime}A]}  
    \infer1[$2-$]{[1,2: -A; 3: -^{\prime}A]}  
    \end{prooftree}
    \quad \text{ or } \quad
    \begin{prooftree}
    \hypo{[1,2,3:A]}
    \infer1[$3-^{\prime}$]{[2:A, 3:A, \sim^{\prime} A]}  
    \infer1[$1-$]{[1:-A;2:A; 3:-^{\prime} A]}  
    \infer1[$2-$]{[1,2:-A; 3: -^{\prime} A]}  
    \end{prooftree} 
\]
Multiplicative connectives:
\[
    \begin{prooftree}
        \hypo{[1,2,3:A]}
        \hypo{[1,2,3:B]}
        \hypo{[1,2,3:A]}
        \hypo{[1,2,3:B]}
        \infer2[3$\wedge^\prime$]{[1,2: A,B; 3:A\wedge^\prime B]}
        \infer3[2$\wedge$]{[1:A,B;2:A\wedge B; 3: A\wedge^\prime B]}
        \infer1[1$\wedge$]{[1,2:A\wedge B; 3: A\wedge^\prime B]}
    \end{prooftree}
\]
\[
    \begin{prooftree}
        \hypo{[1,2,3:A]}
        \hypo{[1,2,3:B]}
        \hypo{[1,2,3:A]}
        \hypo{[1,2,3:B]}
        \infer2[1$\vee$]{[1: A\vee B; 2,3: A,B]}
        \infer3[2$\vee$]{[1:A\vee B; 2:A\vee B; 3: A,B]}
        \infer1[3$\vee^\prime$]{[1,2:A\vee B; 3: A\vee^\prime B]}
    \end{prooftree}
    \]
    \[
    \begin{prooftree}
        \hypo{[1,2,3:A]}
        \hypo{[1,2,3:B]}
        \hypo{[1,2,3:A]}
        \hypo{[1,2,3:B]}
        \infer2[1$\supset$]{[1: A, A\supset B; 2:A,B;3:B]}
        \infer3[2$\supset$]{[1:A, A\supset B; 2:A\supset B; 3: B]}
        \infer1[3$\supset^\prime$]{[1,2:A\supset B; 3: A\supset^\prime B]}
    \end{prooftree}
\]
Additive connectives: see Figure 1
\begin{sidewaysfigure}[htbp]
\fbox{
\parbox{\linewidth}{
\tiny
\[
\begin{prooftree}
    \hypo{[1,2,3:A]}
    \infer1[1$\wedge$1]{[1:A\wedge B;2,3:A]}
    \hypo{[1,2,3:B]}
    \infer1[1$\wedge$2]{[1:A\wedge B;2,3:B]}
    \hypo{[1,2,3:A]}
    \infer1[1$\wedge$1]{[1:A\wedge B;2,3:A]}
    \infer3[2$\wedge$1]{[1:A\wedge B;2:A\wedge B;3: A]}
    \hypo{[1,2,3:A]}
    \infer1[1$\wedge$1]{[1:A\wedge B;2,3:A]}
    \hypo{[1,2,3:B]}
    \infer1[1$\wedge$2]{[1:A\wedge B;2,3:B]}
    \hypo{[1,2,3:B]}
    \infer1[1$\wedge$2]{[1:A\wedge B;2,3:B]}
    \infer3[2$\wedge$2]{[1:A\wedge B;2:A\wedge B;3: B]}
    \infer2[3$\wedge^\prime$]{[1:A\wedge B;2:A\wedge B;3: A\wedge^\prime B]}
\end{prooftree}
\]
\[
\begin{prooftree}
    \hypo{[1,2,3:A]}
    \infer1[3$\vee^{\prime}$1]{[1,2:A;3:A\vee^{\prime} B]}
    \hypo{[1,2,3:B]}
    \infer1[3$\vee^{\prime}$2]{[1,2:B;3:A\vee^{\prime} B]}
    \hypo{[1,2,3:A]}
    \infer1[3$\vee^{\prime}$1]{[1,2:A;3:A\vee^{\prime} B]}
    \infer3[2$\vee$1]{[1:A;2:A\vee B;3: A\vee^{\prime} B]}
    \hypo{[1,2,3:A]}
    \infer1[3$\vee^{\prime}$1]{[1,2:A;3:A\vee^{\prime} B]}
    \hypo{[1,2,3:B]}
    \infer1[3$\vee^{\prime}$2]{[1,2:B;3:A\vee^{\prime} B]}
    \hypo{[1,2,3:B]}
    \infer1[3$\vee^{\prime}$2]{[1,2:B;3:A\vee^{\prime} B]}
    \infer3[2$\vee$2]{[1:B;2:A\vee B;3: A\vee^{\prime} B]}
    \infer2[1$\vee$]{[1:A\vee B;2:A\vee B;3: A\vee^\prime B]}
\end{prooftree}
\]
\[
\begin{prooftree}
    \hypo{[1,2,3:A]}
    \infer1[3$\supset^{\prime}$1]{[2:A;3:A\supset^{\prime} B, A]}
    \hypo{[1,2,3:B]}
    \infer1[3$\supset^{\prime}$2]{[1,2:B;3:A\supset^{\prime} B]}
    \hypo{[1,2,3:A]}
    \infer1[3$\supset^{\prime}$1]{[2:A;3:A\supset^{\prime} B, A]}
    \infer3[2$\supset$1]{[2:A\supset B;3: A\supset^{\prime} B, A]}
    \hypo{[1,2,3:A]}
    \infer1[3$\supset^{\prime}$1]{[2:A;3:A\supset^{\prime} B, A]}
    \hypo{[1,2,3:B]}
    \infer1[3$\supset^{\prime}$2]{[1,2:B;3:A\supset^{\prime} B]}
    \hypo{[1,2,3:B]}
    \infer1[3$\supset^{\prime}$2]{[1,2:B;3:A\supset^{\prime} B]}
    \infer3[2$\supset$2]{[1:B;2:A\supset B;3: A\supset^{\prime} B]}
    \infer2[1$\supset$]{[1:A\supset B;2:A\supset B;3: A\supset^\prime B]}
\end{prooftree}
\]
}}
\caption{proof of Theorem 1---additive connectives}
\end{sidewaysfigure}
\hfill $\blacksquare$
\end{proof}
The usual strategy of proving conservativity via \textsc{Cut}-eliminability/the sub-formula property is unavailable to us as our base calculus $\{\textsc{Cut}, \textsc{Id}\}$ already violates \textsc{Cut}-eliminability. Instead, we will prove conservativity directly. We first show what sequents are provable in the base system (Lemma 1). We then prove that no other connective-free sequents are provable in the extended system (Lemmata 2,3).
\begin{definition}[derivation height $h$]\

    \noindent The height $h$ of a derivation is the maximum number of successive rule applications from a (possibly empty) leaf to its root, i.e. the length of the longest branch.
\end{definition}
\begin{lemma}\

    \noindent For all theorems $\mathfrak{s}$ of $\{\textsc{Cut}, \textsc{Id}\}$ and for all $i \in \mathcal{D}$, $\mathfrak{s}=[1:\Gamma_1;2:\Gamma_2;3:\Gamma_3]$ such that for some $\mathcal{L}_\textbf{Cal}$-formula $A$, each $\Gamma_{i}$ only contains some non-zero number of instances of $A$ .
\end{lemma}
\begin{proof}
    By induction on the height of the proof $\mathfrak{p}$ of $\mathfrak{s}$. We hypothesise that for any proof $\mathfrak{p}^\prime$ of $\mathfrak{s}^{\prime}$ with $h(\mathfrak{p}^\prime)< h(\mathfrak{p})$, the lemma holds for $\mathfrak{s}^{\prime}$. Let $i,j,k\in \mathcal{D}$ be pairwise distinct. If $h(\mathfrak{p})= 1$, $\mathfrak{s}$ is an instance of \textsc{Id}, the lemma trivially holds. If $h(\mathfrak{p})>1$, the final rule application was an instance of the \textsc{Cut} rule. As both \textsc{Cut}-premises have proofs of lower height than $h(\mathfrak{p})$, the inductive hypothesis applies. Hence, the \textsc{Cut}-premises are of shape $[i:\underbrace{A,\dots, A}_{m\text{-times}};j:\underbrace{A,\dots, A}_{n\text{-times}}; k:\underbrace{A,\dots, A}_{o\text{-times}}]$ and
    $[i:\underbrace{A,\dots, A}_{p\text{-times}};j:\underbrace{A,\dots, A}_{q\text{-times}};k:\underbrace{A,\dots, A}_{r\text{-times}}]$, respectively, for pairwise distinct $i,j,k\in \mathcal{D}$. Thus: $$\mathfrak{s}=[i:\underbrace{A,\dots, A}_{(m-1)\text{-times}}, \underbrace{A,\dots, A}_{p\text{-times}};j:\underbrace{A,\dots, A}_{n\text{-times}},\underbrace{A,\dots, A}_{(q-1)\text{-times}};k:\underbrace{A,\dots, A}_{o\text{-times}}, \underbrace{A,\dots, A}_{r\text{-times}}].$$ As $p,n,o,r>0$, the lemma holds.\hfill $\blacksquare$
\end{proof}
\begin{definition}[sequent complexity $c$]\

    \noindent The \textit{complexity} $c$ of a sequent $[1:\Gamma;2:\Delta;3:\Sigma]$ is $|\Gamma|+|\Delta|+|\Sigma|$.
\end{definition}
From Definition 11 and Lemma 1, we get:
\begin{corollary}\

    \noindent For all theorems $\mathfrak{s}$ of $\{\textsc{Cut}, \textsc{Id}\}$, $c(\mathfrak{s})\geq3$.
\end{corollary}
It follows that there are two ways our extended system could violate conservativity: by proving some sequent with complexity $\leq 2$, or by proving some connective-free sequent that contains instances of more than one distinct formula. In turn, we show that neither is the case.\begin{lemma}\

    \noindent Let \# be any multiplicative or additive \textbf{K3} connective with operational rules $\mathfrak{R}_\#$.
    
    For all theorems $\mathfrak{s}$ of $\{\textsc{Cut}, \textsc{Id}\}\cup \mathfrak{R}_\#$, $c(\mathfrak{s})\geq3$.
\end{lemma}
\begin{proof}
    We refer to the connective rules as defined in Tables 2 and 4, including the use of `$s$', `$t$', `$u$' for context sub-sequents.
    
    Proof by induction on the height of the proof $\mathfrak{p}$ of $\mathfrak{s}$. If $h(\mathfrak{p})=1$, $\mathfrak{s}$ is an instance of \textsc{Id} and $c(\mathfrak{s})=3$. If $h(\mathfrak{p})>1$, the last rule applied could be a connective rule or \textsc{Cut}. Assuming the latter, both \textsc{Cut}-premiss sequents $[i:A]s$ and $[j:A]t$ have complexity $\geq 3$ in virtue of the inductive hypothesis. Hence, $c(s)\geq 2$ and $c(t)\geq 2$. As $\mathfrak{s}=st$ and $c(st)\geq 4$, the lemma holds.

    Assume that the last rule applied was a $-$-rule. In virtue of the hypothesis, the premiss sequent of each $-$-rule is of complexity $\geq 3$. By the rule definition, the complexity of the conclusion is also of complexity $\geq 3$ and the lemma holds.
    
    Assume that the last rule application was a multiplicative $\wedge$-rule. The case of $3\wedge$ is equivalent to that of \textsc{Cut}, only with $\mathfrak{s}=[3:A\wedge B]st$ and, thus, $c(\mathfrak{s})\geq 5$. In case of $1\wedge$, $c([1:A\wedge B]s)\geq 3$ since $c(s)\geq 2$. This must be as the premiss sequents $[1:A,B]s$ cannot be an instance of \textsc{Id} and the conclusion of \textsc{Cut} and $3\wedge$ must be of complexity $\geq 4$. The same holds for conclusions of $2\wedge$. Whilst $c([2,3:A]s)\geq 3$ and $c([2,3:B]u)\geq 3$ by our inductive hypothesis, $c([2:A,B]t)\geq 4$. This is as $[2:A,B]t$ cannot be an instance of \textsc{Id}, nor can it be a conclusion of $1\wedge$, unless $c([1:A\wedge B]s)\geq 4$. For if $c([1:A\wedge B]s)= 3$, $c([1:A, B]s)= 4$. This would mean that the $1\wedge$-premiss $[1:A, B]s$ must have been a conclusion of \textsc{Cut} since $3\wedge$ only derives sequents of complexity $\geq 5$. As the conclusion of \textsc{Cut} is of greater complexity than its premises, the latter must have been instances of \textsc{Id}. Hence, $[1:A, B]s=[1:A,A;2,3:A]$ and $[1:A\wedge B]s=[1:A\wedge A;2,3:A]$, which cannot serve as the $3\wedge$-premiss $[2:A,B]t$. Therefore, $c(t)\geq 2$ and $c([stu])\geq 4$. 

    Assume that we last applied an additive $\wedge$-rule. $1\wedge$ does not modify sequent complexity from premiss to conclusion and the lemma trivially holds. In the case of $3\wedge$, $c([3:A]s)\geq 3$ in virtue of the hypothesis. Hence, $c(s)\geq 2$ and $c([3:A\wedge B]s)\geq 3$. In the case of $2\wedge(1)$, $c([2:A]s)\geq 3$ due to our hypothesis. Thus, $c(s)\geq 2$ and $c([2:A\wedge B]s)\geq 3$.

    The cases of multiplicative and additive $\vee$- and $\supset$-rules are analogous to the respective $\wedge$ cases and are left as an exercise to the reader. \hfill $\blacksquare$
\end{proof}
\begin{lemma}\

    \noindent Let \# be any multiplicative or additive \textbf{K3} connective with operational rules $\mathfrak{R}_\#$.
    
    Any sequent $\mathfrak{s}=[1:\Gamma_1;2:\Gamma_2;3:\Gamma_3]$ provable in $\{\textsc{Id, Cut}\}\cup \mathfrak{R}_{\#}$ contains:
        \begin{enumerate}
        \item instances of some primitive formula $A^{pr}$ and nothing else, or
        \item some instance of a \#-formula $\texttt{\#}$ with active formulae $A,B$ s.t.
        \begin{enumerate}
            \item For all $i\in\mathcal{D}, PS(\Gamma_i)=PS(\texttt{\#})$ if $PS(A)=PS(B)$ (including if $\texttt{\#}$ is unary),
            \item $\texttt{\#}\in \Gamma_1$ if $PS(A)\neq PS(B)$, and $\#$ is multiplicative $\vee, \supset$ or additive $\wedge$,
            \item $\texttt{\#}\in \Gamma_3$ if $PS(A)\neq PS(B)$, and $\#$ is multiplicative $\wedge$ or additive $\vee, \supset$.
        \end{enumerate} 
    \end{enumerate}
\end{lemma}
\begin{proof}
    By induction on the height of the proof $\mathfrak{p}$ of $\mathfrak{s}$. If $h(\mathfrak{p})=1$, then $\mathfrak{s}$ is an instance of \textsc{Id} and contains exclusively the same primitive or complex formula in each component.
    
    Let $h(\mathfrak{p})>1$. Assume the last rule applied in $\mathfrak{p}$ was \textsc{Cut} and $\mathfrak{p}$ is of the form: \begin{prooftree}
\hypo{}
\infer[rule code={\hbox{\tikz
    \draw (0,0) -- (\hsize,0) -- (0.5\hsize,-2em) (0.5\hsize,-0.875em) node{\small $\mathfrak{p}_1$}  (0.5\hsize,-2em) -- (0,0);}}]1{[i:C]s}
\hypo{}
\infer[rule code={\hbox{\tikz
    \draw (0,0) -- (\hsize,0) -- (0.5\hsize,-2em) (0.5\hsize,-0.875em) node{\small $\mathfrak{p}_2$}  (0.5\hsize,-2em) -- (0,0);}}]1{[j:C]t}
\infer2[Cut]{st}
\end{prooftree}, with $i,j\in \mathcal{D}, i\neq j$. Since $h(\mathfrak{p}_1)<h(\mathfrak{p}),h(\mathfrak{p}_2)<h(\mathfrak{p})$, the lemma holds for both $[i:C]s$ and $[j:C]t$ in virtue of the inductive hypothesis. If both \textsc{Cut}-premises satisfy disjunct 1, $s,t$ exclusively contain $A^{pr}$, as does $st$. Assume $[i:C]s$ satisfies disjunct 1 and $[j:C]t$ satisfies disjunct 2a. Then, $C=A^{pr}$ and $t$ must contain $\texttt{\#}$, as does $st$. Moreover as $PS(\texttt{\#})=PS(C)=A^{pr}$, the only primitive sub-formula of each formula in $t$ must be $A^{pr}$ as is trivially the case for $s$. Thus, $st$ satisfies 2a. Assume both \textsc{Cut}-premises satisfy disjunct 2a. If $C=\texttt{\#}$, $PS(D)=PS(\texttt{\#})$, for each formula $D$ in $st$. Thus, if some $D$ is a $\#$-formula, disjunct 2a is satisfied. If no $D$ is a $\#$-formula, each $D$ is primitive and $PS(D)=\{D\}$. Hence, all formulae in $st$ are instances of the same primitive formula and disjunct 1 is satisfied. If $C\neq \texttt{\#}$, there is some $\texttt{\#}$ in $s$ and some $\texttt{\#}^{\prime}$ in $t$. Since $\#$ and $\#^{\prime}$ both have the same primitive sub-formulae as $C$, so does each formula in $s,t$ and $st$ satisfies 2a. Assume either \textsc{Cut}-premiss satisfies disjunct 2b/c. If at least one of $s,t$ contains $\texttt{\#}$ in $\Gamma_1$/$\Gamma_3$, so does $st$ and disjunct 2b/c holds. If neither contains $\texttt{\#}$, $C=\texttt{\#}$. However, this case cannot obtain as otherwise $[i:A]\in\Gamma_n$ and $[j:A]\in\Gamma_n$ for $n\in\{1,3\}$. Then, $n=i=j$, which cannot be by the specification of \textsc{Cut}.
    
Now, assume the last rule applied in $\mathfrak{p}$ was some $-$-rule with premiss sequent $\mathfrak{s}^\prime$. If $\mathfrak{s}^\prime$ satisfies disjunct 1, it exclusively contains instances of some primitive $A$. Then, $\mathfrak{s}$ contains one instance of $-A$ and otherwise only instances of $A$, satisfying disjunct 2a. If $\mathfrak{s}^\prime$ instead satisfies disjunct 2a, so does $\mathfrak{s}$. For irrespective of the shape of active formula $A$, $PS(A)=PS(C)$ for all $C$ in $s$ due to the hypothesis and $PS(A)=PS(-A)$.

Assume the last rule was an additive $\wedge$-rule. As $1\wedge$ introduces $A\wedge B$ in $\Gamma_1$ and $\mathfrak{s}$ has the same primitive sub-formulae as its premiss if $PS(A)=PS(B)$, $\mathfrak{s}$ trivially satisfies disjunct 2. In case of $3\wedge$, if at least one premiss, $[3:A]s$ or $[3:B]s$, satisfies disjuncts $1$ or $2a$, $PS(C)=PS(A)=PS(B)$ for all $C$ in $s$. Hence, $PS(C)=PS(A\wedge B)$, and $\mathfrak{s}$ satisfies 2a. If at least one premiss satisfies $2b$, $s$ contains at least one instance of $\texttt{\#}$ in the first component, and $\mathfrak{s}$ satisfies disjunct 2. The argument for $2\wedge$ is identical to that for $3\wedge$. 

Assume the last rule was a multiplicative $\wedge$-rule. Assume it was $2\wedge$. If all premises satisfy 1 or 2a, each formula in each premiss has the same atomic sub-formulae, and thus $[2:A\wedge B]stu$, satisfying 2a. If at least one premiss satisfies 2c, $s$, $t$ or $u$ contain $\texttt{\#}$ in the third component, as does $\Gamma_3$. Thus, $\mathfrak{s}$ satisfies 2c. The same reasoning applies to $1\wedge$ and $3\wedge$ except that $\texttt{\#}$ is always principal in the latter case.

The cases of additive/multiplicative $\vee$ and $\supset$ are analogous to additive/multiplicative $\wedge$ and left as an exercise.
\hfill $\blacksquare$
\end{proof}
\begin{theorem}[conservativity]\

    \noindent Let \# be any multiplicative or additive \textbf{K3} connective with operational rules $\mathfrak{R}_\#$.
    
    \# \textit{conservatively extends} $\{\textsc{Id, Cut}\}$, i.e. all sequents provable in $\{\textsc{Id, Cut}\}\cup \mathfrak{R}_{\#}$ contain `\#', unless they are also provable in $\{\textsc{Id, Cut}\}$.
\end{theorem}
\begin{proof}
    From Lemmata 1-3. \hfill $\blacksquare$
\end{proof}
\begin{corollary}\

    \noindent
    Any multiplicative or additive \textbf{K3} connective \# is \textit{definable} in $\{\textsc{Id, Cut}\}$.
\end{corollary}
\begin{proof}
    From Theorems 1 and 2, and Definition 6. \hfill $\blacksquare$
\end{proof}
By establishing definability in our minimal derivability relation, we have generated the proofs to which we can now apply our inference behaviour function and obtain semantic clauses for the \textbf{K3} connectives.

\section{Results: Connective Meaning in \textbf{K3} and \textbf{LP}}

We now measure the inference-behaviour semantics in our definability proofs to assess which \textbf{K3} connectives share meaning. Based on this, we give the semantic clause for each \textbf{K3} connective with distinct meaning. We show that our results are also immediately applicable to \textbf{LP}, resulting in identical connective meanings in \textbf{K3} and \textbf{LP}. Based on this, we prove that the \textbf{K3}/\textbf{LP}-connectives are meaning-extensions of the connectives in \textbf{LK}.
\begin{theorem}\

    \noindent Each multiplicative and each additive \textbf{K3} connective has a distinct meaning, and \textbf{K3}-negation ($-$) has two distinct meanings.
\end{theorem}
\begin{proof}
    We measure the inference behaviour in the proof of definability relative to $\{\textsc{Id,Cut}\}$. Note that the meta-logical proof of Theorem 2 does not establish any proofs in the object-logic---i.e. prove theorems of $\{1\#,2\#,3\#,\textsc{Id,Cut}\}$---but merely provides an analysis of the form $\{1\#,2\#,3\#,\textsc{Id,Cut}\}$-proofs can(not) take. Thus, it suffices to apply $\texttt{ib}$ to the $\{1\#,2\#,3\#,\textsc{Id,Cut}\}$-proofs in Theorem 1.
    
    Let $\mathfrak{p}$ be the proof of horizontal inter-derivability for connective $\#$. In virtue of Theorem 1, $\#$ and $\#^\prime$ are inter-definable, i.e. their formulae are horizontally inter-derivable. We can simply stipulate $\#:=\#^\prime$.

    If $\#=-$, $\texttt{ib}_\mathfrak{p}^-=\langle\langle[1,2,3],\textsc{Id}\rangle,\langle[1,1,3],1-\rangle,\langle[1,1,3],3-\rangle,\langle[1,2,3],2-\rangle\rangle$, or $\texttt{ib}_\mathfrak{p}^-=\langle\langle[1,2,3],\textsc{Id}\rangle,\langle[2,3,3],3-\rangle,\langle[2,3,3],1-\rangle,\langle[1,2,3],2-\rangle\rangle$.
    
    Let \# be multiplicative. If $\#=\wedge$, $\texttt{ib}_\mathfrak{p}^\wedge=\langle\langle[1,2,3],\textsc{Id}\rangle, \langle[1,2,3],\textsc{Id}\rangle, \langle[1,2,3],\textsc{Id}\rangle,$ $ \langle[1,2,3],\textsc{Id}\rangle, \langle[1,1,2,2,3],3\wedge\rangle,\langle[1,1,2,3],2\wedge\rangle,$ $\langle[1,2,3],1\wedge\rangle\rangle$.

    If $\#=\vee$, $\texttt{ib}_\mathfrak{p}^\vee=\langle\langle[1,2,3],\textsc{Id}\rangle, \langle[1,2,3],\textsc{Id}\rangle, \langle[1,2,3],\textsc{Id}\rangle,$ $ \langle[1,2,3],\textsc{Id}\rangle, \langle[1,2,2,3,3],$\linebreak $ 1\vee\rangle,\langle[1,2,3,3],2\vee\rangle,\langle[1,2,3],3\vee\rangle\rangle$.

    If $\#=\ \supset$, $\texttt{ib}_\mathfrak{p}^\supset=\langle\langle[1,2,3],\textsc{Id}\rangle, \langle[1,2,3],\textsc{Id}\rangle, \langle[1,2,3],\textsc{Id}\rangle, \langle[1,2,3],\textsc{Id}\rangle, \langle[1,1,2,2,3],$\linebreak $1\supset\rangle,\langle[1,1,2,3],2\supset \rangle,$ $\langle[1,2,3],3\supset \rangle\rangle$.

    Let \# be additive. If $\#=\wedge$, $\texttt{ib}_\mathfrak{p}^\wedge=\langle\langle[1,2,3],\textsc{Id}\rangle, \langle[1,2,3],\textsc{Id}\rangle, \langle[1,2,3],\textsc{Id}\rangle,$ $ \langle[1,2,3],\textsc{Id}\rangle, \langle[1,2,3],\textsc{Id}\rangle, \langle[1,2,3],\textsc{Id}\rangle,
    \langle[1,2,3],1\wedge1\rangle, \langle[1,2,3],1\wedge2\rangle, \langle[1,2,3],1\wedge1\rangle, \langle[1,2,3],1\wedge1\rangle, \langle[1,2,3],1\wedge2\rangle, \langle[1,2,3],1\wedge2\rangle, \langle[1,2,3],2\wedge1\rangle,\langle[1,2,3],2\wedge2\rangle,$\linebreak $\langle[1,2,3],3\wedge\rangle\rangle$.

    If $\#=\vee$, $\texttt{ib}_\mathfrak{p}^\vee=\langle\langle[1,2,3],\textsc{Id}\rangle, \langle[1,2,3],\textsc{Id}\rangle, \langle[1,2,3],\textsc{Id}\rangle, \langle[1,2,3],\textsc{Id}\rangle, \langle[1,2,3],$\linebreak $\textsc{Id}\rangle, \langle[1,2,3],\textsc{Id}\rangle,
    \langle[1,2,3],3\vee1\rangle, \langle[1,2,3],3\vee2\rangle, \langle[1,2,3],3\vee1\rangle, \langle[1,2,3],3\vee1\rangle,$\linebreak $ \langle[1,2,3],3\vee2\rangle, \langle[1,2,3],3\vee2\rangle, \langle[1,2,3],2\vee1\rangle,\langle[1,2,3],2\vee2\rangle, \langle[1,2,3],1\vee\rangle\rangle$.

    If $\#=\ \supset$, $\texttt{ib}_\mathfrak{p}^\supset=\langle\langle[1,2,3],\textsc{Id}\rangle, \langle[1,2,3],\textsc{Id}\rangle, \langle[1,2,3],\textsc{Id}\rangle, \langle[1,2,3],\textsc{Id}\rangle, \langle[1,2,3],$\linebreak $\textsc{Id}\rangle, \langle[1,2,3],\textsc{Id}\rangle,
    \langle[2,3,3],3\supset1\rangle, \langle[1,2,3],3\supset2\rangle, \langle[2,3,3],3\supset1\rangle, \langle[2,3,3],3\supset1\rangle,$\linebreak $ \langle[1,2,3],3\supset2\rangle, \langle[1,2,3],3\supset2\rangle, \langle[2,3,3],2\supset1\rangle,\langle[1,2,3],2\supset2\rangle, \langle[1,2,3],1\supset\rangle\rangle$.

    Checking the $\texttt{ib}$ of any two additive/multiplicative connectives $\#_1, \#_2$ amongst the ones we have just considered (including the two different versions of negation), we find that $\texttt{ib}_\mathfrak{p}^{\#_1}\neq \texttt{ib}_\mathfrak{p}^{\#_2}$. Therefore, the theorem holds by virtue of Definition 8. \hfill $\blacksquare$
\end{proof}
We can see that, as was the case for \textbf{LK}, we conflate semantically distinct connectives with different meanings (e.g. additive/multiplicative connectives, two different negations) if we add them to fully structural \textbf{K3}. To help us distinguish connectives in the fine-grained semantic sense from their coarse-grained behaviour in the fully structural calculus, we make use of our semantic language $\mathcal{L}_\textbf{Cal}^S$. In $\mathcal{L}_\textbf{Cal}^S$, each semantically distinct connective has exactly one name in line with the terminology used earlier for the semantics of \textbf{LK}. We call one version of negation `$\neg$' (with $\texttt{ib}_\mathfrak{p}^-$ as given in the first disjunct of the last proof) and the other one `$\sim$' (with $\texttt{ib}_\mathfrak{p}^-$ as given in the second disjunct), and we refer to multiplicative conjunction/disjunction/implication as `$\otimes$'/`$\oplus$'/`$\rightarrow$' and its additive sister as `$\sqcap$'/`$\sqcup$'/`$\rightsquigarrow$'.

We can easily find the semantic clause of each semantically distinct \textbf{K3}-connective by looking at its \textit{inference behaviour profile} $\texttt{ibp}$. Following the same order as in the proof of the preceding theorem: $\texttt{ibp}_\mathfrak{p}^\neg=\{[1,2,3],[1,1,3],\textsc{Id}, 1\neg, 2\neg, 3\neg\}$, $\texttt{ibp}_\mathfrak{p}^\sim=\{[1,2,3],[2,3,3],\textsc{Id}, 1\sim, 2\sim, 3\sim\}$, $\texttt{ibp}_\mathfrak{p}^\otimes=\{[1,2,3],[1,1,2,3],[1,1,2,2,3],\textsc{Id}, 1\otimes,$\linebreak $ 2\otimes, 3\otimes\}$, $\texttt{ibp}_\mathfrak{p}^\oplus=\{[1,2,3], [1,2,3,3],[1,2,2,3,3],\textsc{Id}, 1\oplus, 2\oplus, 3\oplus\}$, $\texttt{ibp}_\mathfrak{p}^\rightarrow=\{[1,2,3],$\linebreak $[1,1,2,3],[1,1,2,2,3],\textsc{Id}, 1\rightarrow, 2\rightarrow, 3\rightarrow\}$, $\texttt{ibp}_\mathfrak{p}^\sqcap=\{[1,2,3],\textsc{Id}, 1\sqcap, 2\sqcap, 3\sqcap\}$, $\texttt{ibp}_\mathfrak{p}^\sqcup=\{[1,2,3],\textsc{Id}, 1\sqcup, 2\sqcup, 3\sqcup\}$, $\texttt{ibp}_\mathfrak{p}^\rightsquigarrow=\{[1,2,3],[2,3,3],\textsc{Id}, 1\rightsquigarrow, 2\rightsquigarrow, 3\rightsquigarrow\}$. 

To find the semantic clause for a connective (cf. \S 3), we must find the operational rules that only have the structural properties needed to obtain the inference behaviour recorded in Theorem 3. 
\begin{example}
    Let us find the semantic clause for $\otimes$. Only $1\otimes$, $2\otimes$, $3\otimes, \textsc{Id}\in \texttt{ibp}^\otimes_\mathfrak{p}$ and, thus, each rule used in the horizontal-interderivability proof $\mathfrak{p}$ is either an operational $\otimes$-rule, or part of our minimal derivability relation. Hence, it suffices to consider possible ways to minimise the $\otimes$-rules whilst preserving $\texttt{ib}^\otimes_\mathfrak{p}$. 
    
    If we were to change \textit{active formulae} compared to the multiplicative $\wedge$-rules, we would also alter $H(\#,V_i)$ for some $V_i\in \mathfrak{p}$ and, \textit{a fortiori}, $\texttt{ib}^\otimes_\mathfrak{p}$. The same reasoning applies if we tried to change the multiplicative treatment of contexts when passing the inference line. Hence, our $\otimes$-rules must have the same active formulae and multiplicative contextualisation as the general multiplicative  $\wedge$-rules of \textbf{K3}.
    
    However, what set-theoretic structure should the contexts be in the semantic clause? $\texttt{ibp}_\mathfrak{p}^\otimes$ tells us that only three different constellations of set theoretic frameworks occur in $\texttt{ib}_\mathfrak{p}^\otimes$: one formula per dimension ($[1,2,3]$), two formulae in the first and one formula in the second and third dimension ($[1,1,2,3]$), and two formulae in the first dimension, two in the second, and one in the third ($[1,1,2,2,3]$). Hence, our semantic clauses need to have space for 1 or 2 formulae in the first, 1 or 2 formulae in the second, and 3 formulae in the third dimension.

    If we combine these constraints and capture them (and only them) in the form of substructural operational rules, we obtain the following semantic clause for $\otimes$:
    {\footnotesize \[ 
    \begin{prooftree}
        \hypo{[1:A,B; 2:C,D; 3:E]}
        \infer1[$1\otimes$]{[1:A\otimes B; 2:C,D; 3:E]}
    \end{prooftree} \qquad \begin{prooftree}
        \hypo{[1:C;2,3:A]}
        \hypo{[1:D;2,3:B]}
        \hypo{[1:C,D;2:A,B;3:E]}
        \infer3[$2\otimes$]{[1:C,D;2:A\otimes B;3:E]}
    \end{prooftree}
    \]
    \[
    \begin{prooftree}
        \hypo{[1:C;2:D;3:A]}
        \hypo{[1:E;2:F;3:B]}
        \infer2[$3\otimes$]{[1: C,E; 2:D,F;3:A\otimes B]}
    \end{prooftree}
    \]}
\end{example} 
In Table 5, the semantic clauses of the remaining \textbf{K3}-connectives are displayed.
\begin{sidewaystable}[htbp]
\fbox{
\parbox{\linewidth}{
    \centering
    \begin{center}
    \scriptsize
    \begin{tabular}{c c c}
    \begin{prooftree}
        \hypo{[1:C;2:D;3:A]}
        \infer1[1$\neg$]{[1:\neg A, C;2:D]}
    \end{prooftree} 
    &
    \begin{prooftree}
        \hypo{[1:C;2:A;3:D]}
        \infer1[2$\neg$]{[1:C; 2:\neg A; 3:D]}
    \end{prooftree}
    &
    \begin{prooftree}
        \hypo{[1:A,C;2:D]}
        \infer1[3$\neg$]{[1:C;2:D;3:\neg A]}
    \end{prooftree}
    \\ \\
        \begin{prooftree}
        \hypo{[2:C; 3:A, D]}
        \infer1[1$\sim$]{[1:\ \sim A; 2:C; 3: D]}
    \end{prooftree} 
    &
    \begin{prooftree}
        \hypo{[1:C;2:A;3:D]}
        \infer1[2$\sim$]{[1:C; 2:\ \sim A; 3:D]}
    \end{prooftree}
    &
    \begin{prooftree}
        \hypo{[1:A;2:C;3:D]}
        \infer1[3$\sim$]{[2:C; 3:\ \sim A, D]}
    \end{prooftree} 
    \\  \\ 
    \begin{prooftree}
        \hypo{[1:A,B; 2:C,D; 3:E]}
        \infer1[$1\otimes$]{[1:A\otimes B; 2:C,D; 3:E]}
    \end{prooftree} & \multicolumn{2}{c}{\begin{prooftree}
        \hypo{[1:C;2,3:A]}
        \hypo{[1:D;2,3:B]}
        \hypo{[1:C,D;2:A,B;3:E]}
        \infer3[$2\otimes$]{[1:C,D;2:A\otimes B;3:E]}
    \end{prooftree}}
    \\ \\
    \multicolumn{3}{c}{\begin{prooftree}
        \hypo{[1:C;2:D;3:A]}
        \hypo{[1:E;2:F;3:B]}
        \infer2[$3\otimes$]{[1: C,E; 2:D,F;3:A\otimes B]}
    \end{prooftree}} 
    \\ \\
    \begin{prooftree}
        \hypo{[1:A; 2:C; 3:D]}
        \infer1[$1\sqcap$1]{[1:A\sqcap B; 2:C;3:D]}
    \end{prooftree} \ 
    \begin{prooftree}
        \hypo{[1:B; 2:C; 3:D]}
        \infer1[$1\sqcap$2]{[1:A\sqcap B; 2:C;3:D]}
    \end{prooftree}
    &
    \multicolumn{2}{c}{
    \begin{prooftree}
        \hypo{[1:C;2:A;3:A, D]}
        \hypo{[1:C;2:B;3:B, D]}
        \hypo{[1:C;2:A;3:D]}
        \infer3[2$\sqcap$1]{[1:C;2:A\sqcap B;3:D]}
    \end{prooftree}}
    \\ \\
    \multicolumn{2}{c}{
    \begin{prooftree}
        \hypo{[1:C;2:A;3:A, D]}
        \hypo{[1:C;2:B;3:B, D]}
        \hypo{[1:C;2:B;3:D]}
        \infer3[2$\sqcap$2]{[1:C;2:A\sqcap B;3:D]}
    \end{prooftree}}
    &
    \begin{prooftree}
        \hypo{[1:C;2:D;3:A]}
        \hypo{[1:C;2:D;3:B]}
        \infer2[$3\sqcap$]{[1: C;2:D; 3:A\sqcap B]}
    \end{prooftree}
    \\ \\
    \begin{prooftree}
        \hypo{[1:A;2:C;3:D]}
        \hypo{[1:B;2:E;3:F]}
        \infer2[$1\oplus$]{[1:A\oplus B;2:C,E;3:D,F]}
    \end{prooftree}
    &
    \multicolumn{2}{c}{\begin{prooftree}
        \hypo{[1,2:A;3:C]}
        \hypo{[1,2:B;3:D]}
        \hypo{[1:E;2:A,B;3:C,D]}
        \infer3[$2\oplus$]{[1:E;2:A\oplus B; 3:C,D]}
    \end{prooftree}
    }
    \\
    \multicolumn{3}{c}{\begin{prooftree}
        \hypo{[1:C;2:D,E;3:A,B]}
        \infer1[$3\oplus$]{[1:C;2:D,E; 3:A\oplus B]}
    \end{prooftree}}
    \\ \\
    \begin{prooftree}
        \hypo{[1:A;2:C;3:D]}
        \hypo{[1:B;2:C;3:D]}
        \infer2[$1\sqcup$]{[1:A\sqcup B;2:C;3:D]}
    \end{prooftree}
    &
        \multicolumn{2}{c}{
    \begin{prooftree}
        \hypo{[1:A, C;2:A;3:D]}
        \hypo{[1:B, C;2:B;3:D]}
        \hypo{[1:C;2:A;3:D]}
        \infer3[2$\sqcup$1]{[1:C;2:A\sqcup B;3:D]}
    \end{prooftree}}
    \\ \\
    \multicolumn{2}{c}{
    \begin{prooftree}
        \hypo{[1:A, C;2:A;3:D]}
        \hypo{[1:B, C;2:B;3:D]}
        \hypo{[1:C;2:B;3:D]}
        \infer3[2$\sqcup$2]{[1:C;2:A\sqcup B;3:D]}
    \end{prooftree}}
    &
    \begin{prooftree}
        \hypo{[1:C;2:D;3:A]}
        \infer1[$3\sqcup$1]{[1:C;2:D;3:A\sqcup B]}
    \end{prooftree} \
    \begin{prooftree}
        \hypo{[1:C;2:D;3:B]}
        \infer1[$3\sqcup$2]{[1:C;2:D; 3:A\sqcup B]}
    \end{prooftree}
    \\ \\
    \begin{prooftree}
        \hypo{[1:C;2:D;3:A]}
        \hypo{[1:B;2:E;3:F]}
        \infer2[1$\rightarrow$]{[1:A\rightarrow B, C;2:D,E;3:F]}
    \end{prooftree}
    &
    \multicolumn{2}{c}{
    \begin{prooftree}
        \hypo{[1:C;2,3:A]}
        \hypo{[1:B,E;2:B;3:D]}
        \hypo{[1:C,E;2:A,B;3:D]}
        \infer3[$2\rightarrow$]{[1:C,E;2:A\rightarrow B; 3:D]}
    \end{prooftree}
    }
    \\ \\
    \multicolumn{3}{c}{
    \begin{prooftree}
        \hypo{[1:A,C;2:D;3:B]}
        \infer1[3$\rightarrow$]{[1:C;2:D;3:A\rightarrow B]}
    \end{prooftree}
    }
    \\ \\
       \begin{prooftree}
        \hypo{[2:D;3:A,C]}
        \hypo{[1:B;2:D;3:C]}
        \infer2[1$\rightsquigarrow$]{[1:A\rightsquigarrow B;2:D;3:C]}
    \end{prooftree}
    &
    \multicolumn{2}{c}{
    \begin{prooftree}
        \hypo{[2:A;3:A,D]}
        \hypo{[1:B;2:B;3:D]}
        \hypo{[1:C;2:A;3:D]}
        \infer3[2$\rightsquigarrow$1]{[1:C;2:A\rightsquigarrow B;3:D]}
    \end{prooftree}}
    \\ \\
    \multicolumn{2}{c}{
    \begin{prooftree}
        \hypo{[1:C;2:A;3:D]}
        \hypo{[1:B,C;2:B;3:D]}
        \hypo{[1:C;2:B;3:D]}
        \infer3[2$\rightsquigarrow$2]{[1:C;2:A\rightsquigarrow B;3:D]}
    \end{prooftree}}
    &
    \begin{prooftree}
        \hypo{[1:A;2:C;3:D]}
        \infer1[3$\rightsquigarrow$1]{[2:C;3:D, A\rightsquigarrow B]}
    \end{prooftree} \
        \begin{prooftree}
        \hypo{[1:B;2:C;3:D]}
        \infer1[3$\rightsquigarrow$2]{[2:C;3:D, A\rightsquigarrow B]}
    \end{prooftree}
    \end{tabular}
\end{center}
}}
    \caption{semantic clauses, \textbf{K3} and \textbf{LP}}
    \label{tab:placeholder}
\end{sidewaystable}

Let us briefly step back and consider what would have happened, had we not imposed our pragmatic context restrictions in Tables 2 and 4 (see $^\ast,^\dag,^\ddag$). Whilst Theorem 3 would still obtain, we would have found different inference behaviour for all connectives except $-$. This is because uniqueness (Theorem 1) would have failed for the multiplicative connectives unless we had added \textsc{Contraction} to our minimal derivability relation. Similarly, we could not have proven uniqueness for the additive connectives, unless \textsc{Weakening} rules were present.\footnote{We at least need some \textit{restricted version} of \textsc{Contraction}/\textsc{Weakening}.}
\begin{example}
    Let $\mathfrak{p}$ be the horizontal interderivability proof for contextually unrestricted multiplicative conjunction $\wedge_{u}$:
    \[
    \begin{prooftree}
        \hypo{[1,2,3:A]}
        \hypo{[1,2,3:B]}
        \hypo{[1,2,3:A]}
        \hypo{[1,2,3:B]}
        \infer2[3$\wedge_{u}^\prime$]{[1,2: A,B; 3:A\wedge_{u}^\prime B]}
        \infer3[2$\wedge_{u}$]{[1:A,A,B,B;2:A\wedge_{u} B; 3: A\wedge_{u}^\prime B]}
        \infer1[1$\wedge_{u}$]{[1:A, B,A\wedge_{u} B;2: A\wedge_{u} B; 3: A\wedge_{u}^\prime B]}
        \infer1[1$\wedge_{u}$]{[1:A\wedge_{u} B,A\wedge_{u} B;2: A\wedge_{u} B; 3: A\wedge_{u}^\prime B]}
        \infer1[\textsc{1C}]{[1,2:A\wedge_{u} B; 3: A\wedge_{u}^\prime B]}
    \end{prooftree}
    \]
    Thus, $\texttt{ib}_\mathfrak{p}^{\wedge_{u}}=\langle\langle[1,2,3],\textsc{Id}\rangle, \langle[1,2,3],\textsc{Id}\rangle, \langle[1,2,3],\textsc{Id}\rangle,$ $ \langle[1,2,3],\textsc{Id}\rangle, \langle[1,1,2,2,3],$\linebreak $3\wedge_{u}\rangle,\langle[1,1,1,1,2,3],2\wedge_{u}\rangle,\langle[1,1,1,2,3],2\wedge_{u}\rangle,\langle[1,1,2,3],2\wedge_{u}\rangle,$ $\langle[1,2,3],1C\rangle\rangle$, and $\texttt{ibp}_\mathfrak{p}^{\wedge_{u}}=\{[1,2,3],[1,1,2,3],[1,1,1,2,3],[1,1,2,2,3],[1,1,1,1,2,3],1\wedge_{u},2\wedge_{u},3\wedge_{u},$\linebreak $\textsc{Id},1C\}$.

    Again, the semantic clause for $\wedge_{u}$ would take the shape of a minimal triple of operational rules. Hence, we must absorb all structural rules that exceed the resources of our minimal derivability relation (here: $1C$) into the operational rules akin to the absorption of operational rules in $\textbf{G3}$-style systems \cite{dragalin1979matematicheskij, kleene1952introduction}. This can be achieved by modifying the $2\wedge_u$-rule such that its pure premises, $[1,2,3:A]$ and $[1,2,3:B]$, share parametric formulae in the first dimension with its mixed premiss, $[1,2:A,B;3:A\wedge^\prime_u B]$. This prevents the duplications of $A,B$ in the first dimension of the conclusion of $2\wedge_u$ as we instead obtain $[1:A,B;2:A\wedge_u B;3:A\wedge^\prime_u B]$. Hence, no more use of $1C$ is required in the proof thus modified.

    This context restriction is exactly the one we imposed in Table 2, turning $\wedge_u$ back into $\wedge$. Hence, the semantic clause of $\wedge_u$ in $\{\textsc{Id, Cut}\}$ is that of $\otimes$. At first glance, this appears counter-intuitive as we measured different \texttt{ib} for $\wedge_u$ than for $\otimes$. However, we did not measure \texttt{ib} for $\wedge_u$ in $\{\textsc{Id, Cut}\}$ but in $\{\textsc{Id, Cut}, 1C\}$. By tracking $\wedge_u$'s \texttt{ib} in its \textit{more-than-minimal} definability proof in $\{\textsc{Id, Cut}, 1C\}$, we can find the context restrictions needed to find its \textit{minimal} definablility proof in our minimal derivability relation, and thus the meaning of $\wedge_{(u)}$, even if unrestricted $\wedge_u$ was not originally definable in $\{\textsc{Id, Cut}\}$.

    The same reasoning applies to the other unrestricted connectives. All relevant \textit{more-than-minimal} definability proofs can be found in the Appendix.
\end{example}
Returning to our main discussion, we find that the \textbf{LP} connectives have the same semantic clauses for the connectives as their \textbf{K3}-counterparts. The 3-dimensional calculus for the Logic of Paradox \textbf{LP} \parencite{multlog2024lp} is proof-theoretically identical to \textbf{K3} at the level of \textit{vertical} derivability properties of connective rules (see Tables 2,4). Their difference only comes into play at the level of \textit{horizontal} derivability: whilst \textbf{K3} only designates the third dimension ($\mathcal{D}^+=\{3\}$), \textbf{LP} designates both the second and third dimension ($\mathcal{D}^+=\{2,3\}$) as we discussed in Example 1. Subsequently, all our previous results about the \textbf{K3}-connectives apply equally to the \textbf{LP}-connectives.\begin{theorem}\

    \noindent
    The connectives in \textbf{K3} and the connectives in \textbf{LP} mean the same.
\end{theorem}
\begin{proof}
    The uniqueness and conservativity proofs for \textbf{LP} can be obtained from Theorem 1 and 2, respectively, by substituting `\textbf{LP}' for `\textbf{K3}' and---in the proof of Theorem 1 only---by replacing all occurrences of `$\#$' in the second dimension of each sequent with `$\#^\prime$' for each \textbf{K3}/\textbf{LP}-connective $\#$. This is to adjust for differences in the designated dimensions of the derivability relation. As $\#=\#^\prime$ in the proof of Theorem 3, Theorem 3 also applies to \textbf{LP}. Therefore by Definition 8, each connective in \textbf{LP} has a counterpart in \textbf{K3} with the same meaning, and vice versa.\hfill $\blacksquare$
\end{proof}
We are now in a position to prove that the \textbf{K3}/\textbf{LP}-connectives extend the meaning of the \textbf{LK}-connectives according to our notion of meaning-extension.
\begin{theorem}\

    \noindent
    Each connective $\#$ in \textbf{K3}/\textbf{LP} extends the meaning of the corresponding connective $\natural$ in \textbf{LK}, where $\#$ and $\natural$ have the same connective symbol.
\end{theorem}
\begin{proof}
    \textbf{K3} and \textbf{LP} are defined on the set of dimensions $\mathcal{D}_\textbf{K3/LP}=\{1,2,3\}$ and \textbf{LK} is defined on $\mathcal{D}_\textbf{LK}=\{1,2\}$. Let $\#^\prime$ be a connective defined on the set of dimensions $\mathcal{D}^\prime=\{1,3\}$. Intuitively, we obtain the semantic clauses for $\#^\prime$ from the semantic clauses for $\#$ by deleting the $2\#$ rule and the second dimension of each sequent in $1\#, 3\#$. Put formally: for each $\#$ and for each $i\in \mathcal{D}^\prime$, the $i^{th}$ operational rule in the semantic clause for $\#^\prime$, $i\#^\prime$, is identical to the corresponding $i^{th}$ operational rule of the semantic clause for $\#$, $i\#$, except that each 3-dimensional sequent $[1:\Gamma_1; 2:\Gamma_2; 3:\Gamma_3]$ in $i\#$ is uniformly replaced with the 2-dimensional sequent $[1:\Gamma_1; 3:\Gamma_3]$ in $i\#^\prime$.

    As close inspection of Tables 5 and 3 reveals, the semantic clause for each $\#^\prime$ is syntactically identical to that of $\natural$ if $\#^\prime$ and $\natural$ have the same connective symbol (or: name). Trivially, $\#^\prime$ and $\natural$ have identical definability proof in $\{\textsc{Id,Cut}\}$, our minimal derivability relation, and, \textit{a fortiori}, the same \textit{inference behaviour} therein. Therefore, $\#^\prime$ and $\natural$ have the same meaning and the theorem holds in virtue of Definition 9.\hfill $\blacksquare$
\end{proof}

Let us now take stock and discuss the implications of these results in the context of our journey towards integrating I-bS into MUltlog.

\section{Discussion: Towards Automated Proof-Theoretic Semantics}
We have shown that I-bS can be extended to MUltlog's mathematical machinery of $n$-dimensional sequent calculi. We demonstrated how this extension can be done without losing I-bS's key feature: characterising the relationship of connective meanings across different logics. This lays the groundwork for integrating I-bS into a MUltlog-style system that generates semantic clauses for arbitrary multi-valued logics. Let us take stock of our results and discuss their suitability for the overarching project as well as the remaining challenges.

We first adjusted the definitional framework of I-bS to accommodate the $>2$-dimensional sequent structures generated by MUltlog in \S\S2,3. We consider these modifications technically organic and philosophically largely uncontroversial. Whilst our speech-act-based multilateralist sequent readings, which we extended to additional primitive speech-acts such as suspending or doubling judgement, might be somewhat contentious, we do not expect it to add much to the existing critiques---and responses thereto---of Restall-style and other bi- and multilateralist techniques in P-tS. Moreover, as we have noted, no technical results in this paper rely on a multilateralist sequent reading beyond motivating and interpretational factors. 

We then applied this expanded definitional framework to a case-study of 3-dimensional \textbf{K3} and \textbf{LP} in \S\S4,5. Note, however, that our definitional apparatus is not limited to 3-dimensional applications but is applicable to any finite-valued calculus generated by MUltlog. We found that, as in the 2-dimensional case \cite{nagler2026inference, nagler2026measuring}, we can map the relationship of connective meanings amongst different calculi. In particular, we found that the connectives in \textbf{K3} and \textbf{LP} have the same meaning. Hence, we showed how I-bS can account for different designated values in the same calculus. Whilst it might come as a surprise to less proof-theoretically inclined readers that a \textit{paracomplete} (\textbf{K3}) and a \textit{paraconsistent} (\textbf{LP}) logic have the same connective meanings, this is no novelty in the P-tS literature. Notably, \cite{hjortland2013logical} also finds and discusses the identity of connective meanings in 3-dimensional sequent calculi for \textbf{K3} and \textbf{LP}. However, unlike in the present discussion, the result in \cite{hjortland2013logical} does not result from a comprehensive semantic and meta-semantic theory of connective meaning but from a \textit{sameness} criterion in the style of \cite{restall2002carnap, restall2014pluralism}, tying meaning-identity to shared active formulae. In fact, motivated by limiting results for such approaches in \cite{dicher2016proof}, I-bS was partially developed as a correction and generalisation of sameness-based methods. In this sense, our result can be seen as systematising and validating the findings in \cite{hjortland2013logical}.

Finally, we proposed a novel notion of meaning-extension. We showcased how meaning-extensions allow for meaning comparisons in calculi with different dimensionality by proving that the \textbf{K3}/\textbf{LP}-connectives extend the meaning of the classical \textbf{LK}-connectives in \S 5. As the notion can be applied to any two calculi defined on different dimensions, it allows us to preserve I-bS's feature of cross-calculus meaning comparisons to any pair of logics that can be generated using MUltlog. Beyond its technical merits, meaning-extension also extends the idea of Belnap-style harmony (conservativity and uniqueness) to inter-dimensional cases: whilst harmony typically constrains (meaning-preserving) definability of connectives relative to a base calculus of the same dimension, meaning-extension constrains (meaning-preserving) definability of higher-dimensional connectives relative to a base calculus of a lower dimension.

The main remaining challenge is the fact that each application of I-bS to a new calculus requires a new definability proof relative to the minimal derivability relation, in which inference behaviour can then be measured. Whilst we have shown in \cite{nagler2026inference} that only a limited number of ($21$) connectives receives semantic clauses in 2-dimensional sequent calculi---relative to a minimal derivability relation of $\{\textsc{Id, Cut}\}$---there is no immediate generalisation of this result to the arbitrarily finite-valued case. If future research supplements this result, giving I-bS for any MUltlog-generated calculus will be reduced to matching its operational rules to the right semantic clause(s).

An alternative and more sophisticated solution would be to implement MUltlog-generated sequent calculi into a suitable automated theorem prover with the aim of also automating the measurement of inference behaviour, including the generation of required substructural definability proofs. This approach to the generation of semantic clauses would likely be a significant endeavour due to computability constraints in several substructural combinations. However, it also has the potential of automating not just the generation of P-tS but also genuinely proof-theoretic semantic reasoning.

In either way, the tandem of I-bS and MUltlog proves to be a promising candidate for crossing the pen-and-paper threshold towards computational automation.
\subsection*{Acknowledgements}
I sincerely thank Greg Restall, Luca Incurvati, and Franz Berto for their invaluable mentorship. Special thanks to Colin Caret, Graham Priest, Albert Visser and the Utrecht \textit{PhD Logic Workshop} for their helpful comments and feedback on earlier versions of this material. This research was supported by the Linda and Gordon Bonnyman Charitable Trust via the University of St Andrews.
\printbibliography 
\section*{Appendix}
Unrestricted multiplicative connectives:
\[
    \begin{prooftree}
        \hypo{[1,2,3:A]}
        \hypo{[1,2,3:B]}
        \hypo{[1,2,3:A]}
        \hypo{[1,2,3:B]}
        \infer2[$3\wedge^\prime$]{[1,2: A,B; 3:A\wedge^\prime B]}
        \infer3[$2\wedge$]{[1:A,A,B,B;2:A\wedge B; 3: A\wedge^\prime B]}
        \infer1[$1\wedge$]{[1:A, B,A\wedge B;2: A\wedge B; 3: A\wedge^\prime B]}
        \infer1[$1\wedge$]{[1:A\wedge B,A\wedge B;2: A\wedge B; 3: A\wedge^\prime B]}
        \infer1[$1$\textsc{C}]{[1,2:A\wedge B; 3: A\wedge^\prime B]}
    \end{prooftree}
\]
\[
    \begin{prooftree}
        \hypo{[1,2,3:A]}
        \hypo{[1,2,3:B]}
        \hypo{[1,2,3:A]}
        \hypo{[1,2,3:B]}
        \infer2[$1\vee$]{[1: A\vee B; 2,3: A,B]}
        \infer3[$2\vee$]{[1:A\vee B; 2:A\vee B; 3: A,A,B,B]}
        \infer1[$3\vee^\prime$]{[1,2:A\vee B; 3: A\vee^\prime B,A,B]}
        \infer1[$3\vee^\prime$]{[1,2:A\vee B; 3: A\vee^\prime B,A\vee^\prime B]}
        \infer1[$3$\textsc{C}]{[1,2:A\vee B; 3: A\vee^\prime B]}
    \end{prooftree}
    \]
    \[
    \begin{prooftree}
        \hypo{[1,2,3:A]}
        \hypo{[1,2,3:B]}
        \hypo{[1,2,3:A]}
        \hypo{[1,2,3:B]}
        \infer2[$1\supset$]{[1: A, A\supset B; 2:A,B;3:B]}
        \infer3[$2\supset$]{[1:A,A, A\supset B; 2:A\supset B; 3: B,B]}
        \infer1[$3\supset^\prime$]{[1:A, A\supset B; 2:A\supset B; 3: A\supset^\prime B, B]}
        \infer1[$3\supset^\prime$]{[1:A\supset B; 2:A\supset B; 3: A\supset^\prime B, A\supset^\prime B]}
        \infer1[$3$\textsc{C}]{[1,2:A\supset B; 3: A\supset^\prime B]}
    \end{prooftree}
\]
Unrestricted additive connectives: see Figure 2.\footnote{Some of the displayed instances of \textsc{Weakening} are replaceable by \textsc{Cut} and \textsc{Id} (e.g. 
\begin{prooftree}
    \hypo{[1,2:A;3:A\vee^{\prime} B]}
    \hypo{[1,2,3:A]}
    \infer2[\textsc{Cut}]{[1:A,A;2:A;3:A\vee^{\prime} B]}
\end{prooftree}
), but others are not (e.g.  
\begin{prooftree}
    \hypo{[1:A\wedge B;2,3:A]}
    \infer1[$3$\textsc{W}]{[1:A\wedge B;2:A;3:A,B]}
\end{prooftree}).}
\begin{sidewaysfigure}[htbp]
{\tiny\fontsize{5}{5}
\fbox{
\parbox{\linewidth}{
\[
\begin{prooftree}
    \hypo{[1,2,3:A]}
    \infer1[$1\wedge1$]{[1:A\wedge B;2,3:A]}
    \infer1[\textsc{3W}]{[1:A\wedge B;2:A;3:A,A]}
    \hypo{[1,2,3:B]}
    \infer1[$1\wedge2$]{[1:A\wedge B;2,3:B]}
    \infer1[$3$\textsc{W}]{[1:A\wedge B;2:B;3:A,B]}
    \hypo{[1,2,3:A]}
    \infer1[$1\wedge1$]{[1:A\wedge B;2,3:A]}
    \infer3[$2\wedge1$]{[1:A\wedge B;2:A\wedge B;3: A]}
    \hypo{[1,2,3:A]}
    \infer1[$1\wedge1$]{[1:A\wedge B;2,3:A]}
    \infer1[$3$\textsc{W}]{[1:A\wedge B;2:A;3:A,B]}
    \hypo{[1,2,3:B]}
    \infer1[$1\wedge2$]{[1:A\wedge B;2,3:B]}
    \infer1[$3$\textsc{W}]{[1:A\wedge B;2:B;3:B,B]}
    \hypo{[1,2,3:B]}
    \infer1[$1\wedge2$]{[1:A\wedge B;2,3:B]}
    \infer3[$2\wedge2$]{[1:A\wedge B;2:A\wedge B;3: B]}
    \infer2[$3\wedge^\prime$]{[1:A\wedge B;2:A\wedge B;3: A\wedge^\prime B]}
\end{prooftree}
\]
\[
\begin{prooftree}
    \hypo{[1,2,3:A]}
    \infer1[$3\vee^{\prime}1$]{[1,2:A;3:A\vee^{\prime} B]}
    \infer1[$1$\textsc{W}]{[1:A,A;2:A;3:A\vee^{\prime} B]}
    \hypo{[1,2,3:B]}
    \infer1[$3\vee^{\prime}2$]{[1,2:B;3:A\vee^{\prime} B]}
    \infer1[\textsc{1W}]{[1:A,B;2:B;3:A\vee^{\prime} B]}
    \hypo{[1,2,3:A]}
    \infer1[$3\vee^{\prime}1$]{[1,2:A;3:A\vee^{\prime} B]}
    \infer3[$2\vee1$]{[1:A;2:A\vee B;3: A\vee^{\prime} B]}
    \hypo{[1,2,3:A]}
    \infer1[$3\vee^{\prime}1$]{[1,2:A;3:A\vee^{\prime} B]}
    \infer1[$1$\textsc{W}]{[1:A,B;2:A;3:A\vee^{\prime} B]}
    \hypo{[1,2,3:B]}
    \infer1[$3\vee^{\prime}2$]{[1,2:B;3:A\vee^{\prime} B]}
    \infer1[$1$\textsc{W}]{[1:B,B;2:B;3:A\vee^{\prime} B]}
    \hypo{[1,2,3:B]}
    \infer1[$3\vee^{\prime}2$]{[1,2:B;3:A\vee^{\prime} B]}
    \infer3[$2\vee2$]{[1:B;2:A\vee B;3: A\vee^{\prime} B]}
    \infer2[$1\vee$]{[1:A\vee B;2:A\vee B;3: A\vee^\prime B]}
\end{prooftree}
\]
\[
\begin{prooftree}
    \hypo{[1,2,3:A]}
    \infer1[$3\supset^{\prime}1$]{[2:A;3:A\supset^{\prime} B, A]}
    \infer1[$3$\textsc{W}]{[2:A;3:A\supset^{\prime} B, A, A]}
    \hypo{[1,2,3:B]}
    \infer1[$3\supset^{\prime}2$]{[1,2:B;3:A\supset^{\prime} B]}
    \infer1[$3$\textsc{W}]{[1,2:B;3:A\supset^{\prime} B, A]}
    \hypo{[1,2,3:A]}
    \infer1[$3\supset^{\prime}1$]{[2:A;3:A\supset^{\prime} B, A]}
    \infer3[$2\supset1$]{[2:A\supset B;3: A\supset^{\prime} B, A]}
    \hypo{[1,2,3:A]}
    \infer1[$3\supset^{\prime}1$]{[2:A;3:A\supset^{\prime} B, A]}
    \infer1[$1$\textsc{W}]{[1:B;2:A;3:A\supset^{\prime} B, A]}
    \hypo{[1,2,3:B]}
    \infer1[$3\supset^{\prime}2$]{[1,2:B;3:A\supset^{\prime} B]}
    \infer1[$1$\textsc{W}]{[1:B,B;2:B;3:A\supset^{\prime} B]}
    \hypo{[1,2,3:B]}
    \infer1[$3\supset^{\prime}2$]{[1,2:B;3:A\supset^{\prime} B]}
    \infer3[$2\supset 2$]{[1:B;2:A\supset B;3: A\supset^{\prime} B]}
    \infer2[$1\supset$]{[1:A\supset B;2:A\supset B;3: A\supset^\prime B]}
\end{prooftree}
\]
}}}
\caption{proof of Theorem 1---unrestricted additive connectives}
\end{sidewaysfigure}

\end{document}